\documentclass[amsmath,amssymb,aps,twocolumn]{revtex4-2}

\usepackage{graphicx}
\usepackage[caption=false]{subfig}
\usepackage{dcolumn}
\usepackage{siunitx}
\usepackage{bm}
\usepackage{soul}
\usepackage{xcolor}
\usepackage{tikz}
\usepackage{hyperref}
\usetikzlibrary{quantikz}

\usepackage{mathtools}
\usepackage{amsthm}
\newtheorem{thm}{Theorem}
\DeclarePairedDelimiter{\ceil}{\lceil}{\rceil}

\usepackage{braket}

\usepackage{changes}
\definechangesauthor[name=SP, color=purple]{SP} % Annotations from SP
\definechangesauthor[name=WP, color=violet]{WP} % Annotations from WP

\begin{document}

\title{Fault-tolerant resource estimate for quantum chemical simulations: Case study on Li-ion battery electrolyte molecules}
\author{Isaac H. Kim}
\thanks{IK's contribution was carried out while he was affiliated with PsiQuantum. IK’s current affiliation is the Department of Computer Science, UC Davis, Davis, CA 95616}
\affiliation{PsiQuantum, Palo Alto}
\author{Eunseok Lee}
\email{eunseok.s.lee@gmail.com}
\affiliation{Mercedes-Benz Research and Development North America, Sunnyvale, CA 94085, USA}
\author{Ye-Hua Liu}
\affiliation{PsiQuantum, Palo Alto}
\author{Sam Pallister}
\affiliation{PsiQuantum, Palo Alto}
\author{William Pol}
\email{wpol@psiquantum.com}
\affiliation{PsiQuantum, Palo Alto}
\author{Sam Roberts}
\affiliation{PsiQuantum, Palo Alto}

\date{\today}

\begin{abstract}
We estimate the resources required in the fusion-based quantum computing scheme to simulate electrolyte molecules in Li-ion batteries on a fault-tolerant, photonic quantum computer. We focus on the molecules that can provide practical solutions to industrially relevant problems. Certain fault-tolerant operations require the use of single-qubit ``magic states" prepared by dedicated ``magic state factories” (MSFs). Producing and consuming magic states in parallel is typically a prohibitively expensive task, resulting in the serial application of fault-tolerant gates. However, for the systems considered, the MSF constitutes a negligible fraction of the total footprint of the quantum computer, allowing for the use of multiple MSFs to produce magic states in parallel. We suggest architectural and algorithmic techniques that can accommodate such a capability. We propose a method to consume multiple magic states simultaneously, which can potentially lead to an order of magnitude reduction in the computational runtime without additional expense in the footprint.
\end{abstract}

\maketitle

\section{Introduction}
Quantum computers are capable of carrying out computational tasks that are intractable for classical computers, such as integer factorization~\cite{Shor1994}, combinatorial optimization~\cite{Pagano2020}, simulation of interacting quantum many-body systems~\cite{Feynman1982,Lloyd1996}, and quantum chemistry~\cite{QChemistry} to name a few.

One of the most anticipated applications of quantum computers is first-principles quantum chemical simulation of complex molecules~\cite{Muller2015}. While such studies can be performed on classical computers as well, the computational cost typically scales exponentially unless one adopts certain assumptions and approximations~\cite{DFT_benchmark1,DFT_benchmark2}. Quantum computers are expected to enable quantum chemical simulation from first-principles --- with fewer assumptions and approximations --- and thus accurately model the properties of various molecules. In turn, such studies may lead to new insights into those molecules that are otherwise difficult to obtain.

Recently, there have been continued innovative developments in quantum algorithms for quantum chemistry~\cite{Poulin2014,Reiher2016,Babbush2018,Berry2019,DoubleFactorized_MSFT,Lee2020,QC_sulfur,gujarati2021heuristic,Tyler2020,first_quant}, which have led to many orders of magnitude reduction in the resources required to simulate molecules at accuracies beyond the reach of classical computation. These developments have been driven by a focus on a few molecular systems that take part in complex reactions that underlie high-impact open problems in chemistry (\emph{e.g.,} elucidating the mechanism of nitrogen fixation at FeMoco active sites~\cite{Reiher2016,Li_FeMoco}). Importantly, the cost reduction techniques found in these studies are expected to be applicable to other chemical systems as well.

Thus, it is of interest to understand the cost of these state-of-the-art algorithms when applied to a broader class of industrially relevant quantum chemistry problems and their corresponding molecular systems. In this article, we propose one such concrete application: the study of electrolyte molecules for Li-ion batteries. In conventional Li-ion batteries, the liquid electrolyte plays an important role in determining the performance and the stability of the batteries. Quantum chemical simulation of the underlying constituent molecules can help us better understand electrochemical reactions occurring in the liquid electrolytes, and serve as a useful guide in establishing design principles for better electrolytes.

Realizing the promise of quantum computation to accurately simulate complex chemical systems like Li-ion battery electrolyte chemistries will require devices capable of executing hundreds of billions of operations on thousands of logical qubits. Currently, a huge technological effort is underway to build a fault-tolerant quantum computer. As attention in the field turns toward engineering noise-resilient devices capable of performing such large computations, it is increasingly important to move beyond hardware-agnostic resource estimates in terms of qubits and gates, and account for the overhead of fault-tolerance. 

In that spirit, we first describe resource estimates in terms of architecture-agnostic parameters like the number of qubits and gates, followed by compilation of these hardware-agnostic parameters into the standard primitives of our fault-tolerant architecture, culminating in architecture-specific size and runtime estimates. The fully-compiled resource estimate is markedly different to prior studies, in that it makes use of the recently-introduced \emph{fusion-based quantum computation} (FBQC) scheme~\cite{FBQCpaper,Interleaving}.

FBQC is a quantum computing paradigm in which the fundamental building blocks of the computation are resource states; joint, entangling measurements called \emph{fusions}; and feed-forward operations. For a conceptual overview of FBQC, see Appendix \ref{sec:fusion}. A promising technology that can realize these operations in practice is photonics~\cite{FBQCpaper}. Physically, the resource states can be created using a \emph{resource state generator} (RSG), a gadget that emits the resource states encoded in a finite number of photons~\cite{FBQCpaper}. The requisite fusion measurements can be performed in a variety of ways using only linear optical elements and photon-detection~\cite{Browne2005,FBQCpaper}.

A salient feature of the photon-based FBQC scheme is the possibility of \emph{interleaving}~\cite{Interleaving}. Photons can be stored in an optical delay line of large length without significantly degrading the quality of the photonic qubits. This leads to a compelling ability to easily trade off between the footprint of the device and the runtime of the computation, potentially leading to a more modest engineering overhead in building a large-scale quantum computer. Naturally, the overall cost of the algorithm will greatly depend on the length of the delay line.

\begin{figure*}[t]
	\includegraphics[width=0.45\textwidth]{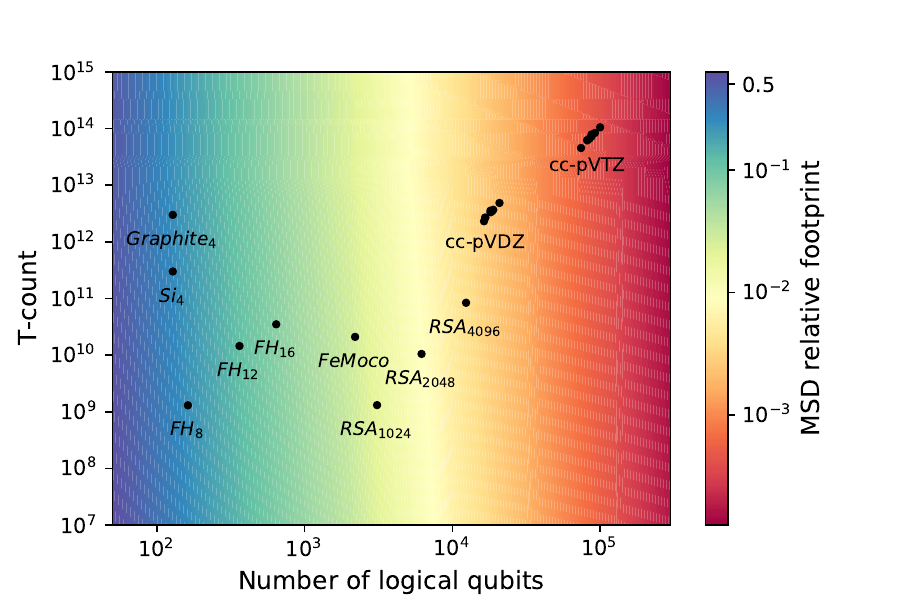} \quad
	\includegraphics[width=0.52\textwidth]{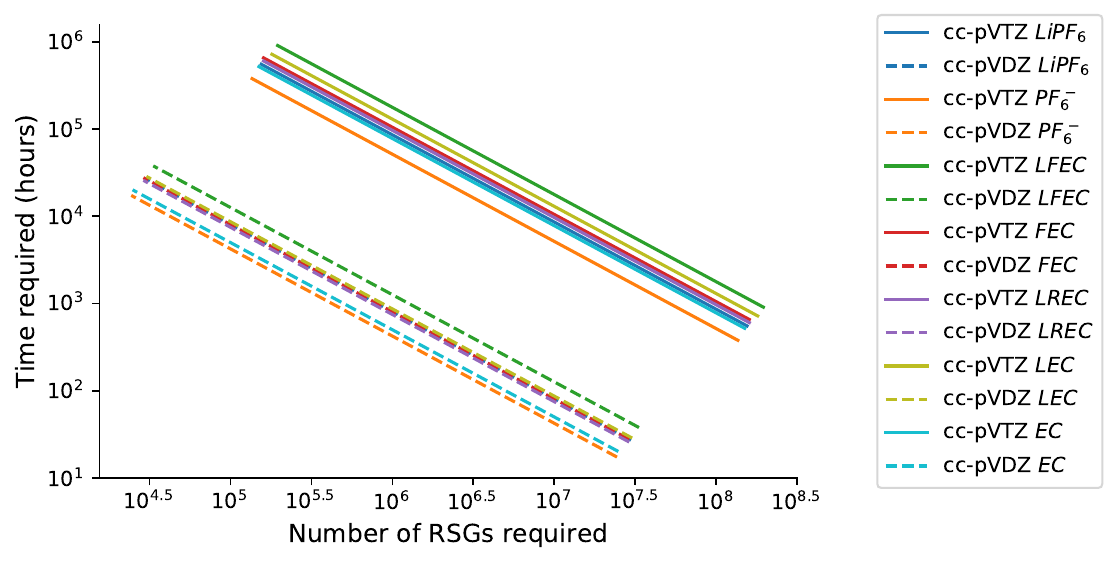}
	\caption{(left) Ratio of magic state distillation (MSD) footprint to total computational footprint for different numbers of logical qubits and T-count. Footprint is measured in terms of number of RSGs required. We assume a linear data block with two multi-level $15$-to-$1$ factories as depicted in Fig. \ref{fig:arch_schematic}. For comparison, we plot the resource estimates for other algorithms, such as: simulation of the Fermi-Hubbard model \cite{campbell2020}; of crystalline materials \cite{Kivlichan2020}; of FeMoco \cite{Lee2020}; and breaking RSA encryption \cite{gidney2019RSA}. Where necessary, resource estimates for quantum chemistry algorithms have been amended to produce eigenenergies to within chemical accuracy. We assume logical error rate parameters $A=0.45$, $B = 1.35$ (i.e. the average of the two regimes considered in Section~\ref{sec:overhead}), with magic states initially prepared with a logical Pauli error rate of $0.1\%$. (right) Footprint and time estimates to perform fault-tolerant quantum computations for various molecules in the cc-pVTZ and cc-pVDZ bases based on Table~\ref{table:double_tcount}. Using interleaving, we can linearly trade-off space and time resources, as displayed. We plot resource estimates for a range of interleaving ratios between $1$  and $1000$, as such ratios can be achieved introducing negligibly higher error rates~\cite{Interleaving}. We assume logical error rate parameters $A=0.5$, $B = 1.6$, with magic states initially prepared with a logical Pauli error rate of $0.1\%$.}
	\label{fig:MSD_footprint_data}
\end{figure*}
Taking into account the distinctive features of FBQC, we perform a detailed analysis of the resources required to simulate a variety of molecules relevant to Li-ion electrolyte chemistry on a fault-tolerant quantum computer. The result of this analysis is succinctly summarized in Fig. \ref{fig:MSD_footprint_data}.

The scope of this article is really two-fold: after first providing concrete, fully-compiled resource estimates on the cost of simulating the aforementioned chemical systems, we study the degree to which considerations of parallelization at both the algorithmic and architectural level can reduce the overall resource requirements to perform these algorithms. Certain fault-tolerant operations require the use of single-qubit ``magic states" prepared by dedicated ``magic state factories” (MSFs)~\cite{Bravyi2005}. According to our resource estimates, having an MSF that produces magic states serially is suboptimal for simulating the systems under consideration. For small quantum algorithms (\emph{i.e.,} few qubits and gates), the footprint of the factory constitutes a substantial fraction of the quantum computer; see Ref.~\cite{Gidney2019efficientmagicstate,litinski2019game,litinski2019magic} for recent attempts to optimize the factory. However, as the number of logical qubits needed for the algorithm increases, the relative size of the MSF will become smaller. In that regime, the (relative) extra cost of adding another magic state factory will be small; while it has been known that this will happen eventually as the number of logical qubits increases (for instance, recent studies~\cite{Lee2020,gidney2019RSA} have considered multiple MSFs in their estimates), the exact point where this crossing occurs has remained unaddressed in the literature. We fill this gap by performing a detailed study; the answer, which depends on the number of qubits and the T-gate count, is plotted in Fig. \ref{fig:MSD_footprint_data}. Our study suggests that there are quantum algorithms with practical applications that lie outside the oft-assumed regime in which the size of the MSF is relatively large.

For these ``intermediate''-sized fault-tolerant quantum computations, it is beneficial to use multiple MSFs, and to use them efficiently. We provide algorithmic and architectural adjustments that can accommodate the injection of multiple distilled magic states whilst making the Clifford gate count relatively small. On the algorithmic front, we perform a detailed analysis on the cost of parallelizing some of the key subroutines in quantum chemistry simulation algorithms. While many of the results we summarize in this paper are well-known, some are new. For example, we introduce a novel method to apply a specific sequence of Givens rotations used in fermionic basis change (\emph{i.e.,} ``Gizens'' rotations) in logarithmic depth (a proof is provided in Appendix \ref{sec:gizens_proof}). On the architectural side, our work builds upon the scheme of Litinski~\cite{litinski2019game}, which is an architecture that can implement an arbitrary quantum algorithm using $n_T$ T-gates in $\mathcal{O}(n_Td)$ time, where $d$ is the code distance of the individual logical qubits, independent of the number of Clifford gates in the algorithm. We propose a further optimization of this architecture, leading to a computation time of $\mathcal{O}(n_Td/m)$ for a variable parameter $m$ while maintaining the code distance without increasing the footprint, assuming that (i) $m$ is small compared to $d$ and (ii) that the circuit is structured in such a way that $m$ T-gates can be applied simultaneously in the same time step.

We emphasize that while the explicit resource estimates of RSG count and runtime only apply to the FBQC paradigm, the results in other sections (particularly on the size of the magic state factories) apply more generally to \textit{any} surface code-based architecture. The parallelization techniques detailed in this article may reduce the runtime requirements by an order of magnitude; however, while the numerical results were based on a Monte Carlo simulation, the performance of our new architectural proposal that uses multiple MSFs has not been simulated in detail. We leave this to future studies.

\section{Materials and Methods}
\label{sec:modeling}
We first describe the Li-ion battery chemistry use-case and the molecules under consideration and detail the classical computational pipeline required to specify the problem and the quantities of interest, then we describe the quantum algorithms we use to calculate the quantities of interest and quantify the resources required to implement these algorithms using a fault-tolerant cost model. We then detail how we convert these hardware-agnostic estimates into architecture-specific size and runtime estimates.

\subsection{Battery chemistry use-case}
\label{sec:chem_background}
In conventional Li-ion batteries, the liquid electrolyte plays an important role in determining the performance and stability of the batteries~\cite{Goodenough2010,Xu2004,Weber2019}. Electrolytes consist of three components; solvent molecules like ethylene carbonate (EC), conducting-salt molecules like LiPF$_6$, and additive molecules like fluoroethylene carbonate (FEC). The solvent is the host for the salt and the additive, and thus needs to have high solvation, low viscosity, inertness, and safety. Conducting-salts dissociate into ions in the solvent, and these ions carry electricity and thus need to have sufficient solubility into solvent and high stability against side reactions with other chemical components. Additives are substances added to electrolytes to enhance specifically targeted properties, and need to provide durable performance in the long term.

Quantum chemical simulation can be used to understand the electrochemical reactions that occur in these constituent molecules, and thus aid in designing better electrolytes~\cite{Urban2016,Wang2018,Zhang2006}. As an example, let us assume that we aim to elaborate the effect of additives to Li-ion solvation~\cite{Im2017}. Although the entire process involves several intermediate steps and may evolve along different reaction paths, the Li-ion solvation ultimately requires Li to dissociate from salt molecules. Quantum chemical simulation can be used to calculate the dissociation energy of the salt molecules in two different environments: one with solvent-only, and another with solvent as well as additives. In the classical chemical regime, we are interested in determining the \textit{reaction rate}, a quantitative proxy for the relative dominance of two chemical reactions. This quantity is proportional to the Boltzmann factor of the change in enthalpy of the two reactions. Thus, the calculated Li-dissociation energies can be used to determine which environment is more favorable for both the Li-dissociation as well as the solvation, quantitatively.

While quantum chemical simulation can be carried out from first-principles (solving Schr\"odinger's equation without approximations), it is common to adopt approximations and assumptions because of the high computational expense originating from the exponentially-scaling complexity in a first-principles approach. Although such approximations and assumptions have been practically applied to a variety of chemical studies at a compromised accuracy, there is a fundamental limit in both the accuracy of simulation results and the spectrum of chemical systems that can be simulated. In particular, as the focus of chemical studies continues to shift to smaller-scales and evolves to include more comprehensive descriptions of the environments the chemical systems under study inhabit, demand for more precise simulations consistently increases. Quantum computers are believed to mitigate the exponentially-scaling computational complexities in certain problems and thus enable quantum chemical simulations in a full configuration-interaction scheme.

In this study, we take EC, LiPF$_6$, and FEC to represent the solvent, conducting-salt, and additive molecules, respectively. We estimate the computational resources required to study their interactions with Li ions via quantum chemical simulation. Our proposed scenario is to perform the geometric optimization using density functional theory (DFT) calculations classically, and then calculate the energy via quantum computation. We consider quantum phase estimation (QPE) for the quantum part of the computation; in particular, we consider recently developed quantum algorithms~\cite{DoubleFactorized_MSFT,Lee2020}.

\subsubsection{Chemical accuracy}
Here we comment on the required accuracy for the energy calculations. Ultimately, the quantity we are really interested in is the reaction rate, which is proportional to $e^{-\beta \Delta E}$, where $\beta$ is the inverse temperature and $\Delta E = E_1 - E_2$ is the difference between two energies $E_1$ and $E_2$. When the uncertainty in $\Delta E$ is $\approx 1.36$\;kcal/mol at room temperature ($\beta^{-1} \approx 0.5922$\;kcal/mol), computed reaction rates and experimentally determined reaction rates differ by an order of magnitude. For this reason, the desired so-called ``chemical" accuracy for calculating $\Delta E$ is historically taken to be approximately $1$\;kcal/mol~\cite{nobel_chem}. We further note that while one would like to calculate $\Delta E$, we cannot calculate this difference directly; the algorithm at our disposal (QPE) computes \textit{individual} energies $E_1$ and $E_2$, \textit{not} $\Delta E$. To achieve chemical accuracy in the calculation of $\Delta E$, but given only the ability to calculate individual energies, one must in principle compute $E_1$ and $E_2$ to chemical accuracy due to propagation of errors. For instance, suppose $E_{2_\text{calc}} = E_{2_\text{true}}+ E_{2_\text{error}}$ at the initial state of a chemical reaction and $E_{1_\text{calc}} = E_{1_\text{true}}+ E_{1_\text{error}}$ at the transition state. Then, the energy difference is calculated to be $E_{2_\text{calc}} - E_{1_\text{calc}} = E_{2_\text{true}}+ E_{2_\text{error}} - E_{1_\text{true}} -  E_{1_\text{error}}$. It is common to assume that $E_{2_\text{error}}$ is similar to $E_{1_\text{error}}$, and so cancel them out; this is a reasonable assumption to make for some molecules (mostly oxides and fluorides) but it is not for the other molecules. While our battery chemistry use-case may not explicitly require the calculation of the individual energies to chemical accuracy, as we mention above, we aim to push the precision achievable with computational methods via quantum computation to match (or even exceed) experimental accuracy, and so we opt to calculate each individual energy to great accuracy, namely $1$\;mHartree. 

\subsubsection{Active spaces}
Here we note that unlike most studies on quantum chemical simulation via quantum computation, we do not consider any active spaces, \emph{i.e.}, we consider all electrons, including core electrons. We note that the core electrons are not absolutely necessary to capture the chemistry we are interested in. Further, it would not necessarily be appropriate to choose an active space since these molecules do not exhibit strong correlations. However, the aim is to relax several assumptions and approximations typically made in classical computational methods, and to estimate the cost of QPE when pushing the boundary of quantum chemical simulations for molecules - applying some of the hardest conditions for the simulability of important quantum chemical problems (i.e. including all electrons, including core electrons that one may not typically treat).

\subsection{Molecules and computational description}
\label{sec:comp_description}
Here we briefly describe the computational specification of the electrolyte molecules under study. The structures of EC, PF$_6^-$, and FEC molecules mentioned above in Section \ref{sec:chem_background} are illustrated in Fig. \ref{fig:molecules}. Note that PF$_6^-$ is in an anodized state while the others are in a neutral state. In addition to these molecules, their variants with additional Li attached (LEC, LiPF$_6$, and LFEC, respectively) and EC-variant with additional Li substituted (LREC) are also examined considering potential reactions of electrolyte molecules with Li-ions. The geometric optimization of these molecules is processed via a Density Functional Theory (DFT) calculation to obtain the optimal positions of the atoms. The generalized gradient approximation is applied using the Perdew-Burke-Ernzerhof parametrization~\cite{PBE1996}, as implemented in the Vienna Ab-initio Software Package~\cite{VASP1,VASP2,VASP3,VASP4}. A cutoff energy of $520$\;eV is used and the $k$-point mesh was adjusted to ensure convergence of $1$\;meV per atom. The molecules are placed in an empty $15 \times 15 \times 15(\text{\r{A}}^3)$ supercell, and the volume and shape of the supercell are fixed during the relaxation. See Tables \ref{table:geometry_one} and \ref{table:geometry_two} for the atomic coordinates of these molecules after geometric optimization. Note that the plane wave basis set is used only for the purpose of the geometric optimization. When the computational resources required for the quantum computation are estimated, Gaussian-type basis sets are considered.

\begin{table}[h]
	\renewcommand{\arraystretch}{1.15}
	\begin{tabular}{c c | c | c | c}
		Molecule /& Atom & x & y & z \\
		\hline
		\textbf{EC} & O &  -1.5282  &  -1.8372 &  -1.3201 \\
		& O   & -1.2811  &  0.2183  & -0.3662 \\
		&O   & 0.4549 &  -1.2114 &  -0.3859 \\
		&C   & -0.8513 &   -1.0235 &  -0.7469 \\
		&C  & -0.2205 &   0.9382  &  0.2895 \\
		&C  &  0.9440 &  -0.0685  &  0.3405 \\
		&H  & -0.5731 &   1.2401 &    1.2848 \\
		&H  &  0.0094 &   1.8361 &  -0.3014 \\
		&H  &  1.1929 &  -0.3884  &  1.3622 \\
		&H   & 1.8529 &   0.2962 &  -0.1565 \\
		\hline
		\textbf{LEC} & O  &  0.4600  & -1.1954 &  -0.3751 \\
		&O  & -1.2683  &  0.2344 &  -0.3527 \\
		&O &  -1.5180  & -1.8063 &  -1.2870\\
		&C   & 0.9670&   -0.0350  &  0.3631\\
		&C  & -0.1914 &   0.9686  &  0.3147\\
		&C  & -0.8205 &  -0.9776 &  -0.7215\\
		&H  &  1.8740 &   0.3072 &  -0.1501\\
		&H  &  1.2114 &  -0.3775  &  1.3780\\
		&H  &  0.0215&    1.8593 &  -0.2924\\
		&H  & -0.5661  &  1.2627 &   1.3024\\
		&Li  &-2.3899   &-3.0578&   -2.2421 \\
		\hline
		\textbf{LREC} & O &   0.4948 &  -1.0310 &  -0.6892 \\
		&O  & -1.2021  &  0.1934 &   0.1259\\
		&O &  -1.6545 &  -1.6749  & -1.1148\\
		&C &   1.0688 &  -0.0436  &  0.2773\\
		&C  &  0.0565&    1.0553 &   0.3496\\
		&C &  -0.8519 &  -0.9225  & -0.6173\\
		&H  &  2.0502 &   0.2691 &  -0.1020\\
		&H &   1.2012 &  -0.6203 &   1.2084\\
		&H   & 0.1129   & 1.6555 &  -0.5769\\
		&Li & -1.3938  &  1.2441  &  1.6320\\
		\hline
		\textbf{FEC} & F  &  2.0238  &  0.4340 &  -0.2441\\
		&O &   0.4384 &  -1.2159  & -0.3308\\
		&O&   -1.3137 &   0.2050 &  -0.3844\\
		&O &  -1.4938 &  -1.8355 &  -1.3896\\
		&C &   0.8931  & -0.0880 &   0.3799\\
		&C &  -0.2719  &  0.9089 &   0.3131\\
		&C  & -0.8619  & -1.0207 &  -0.7745\\
		&H &   1.2057 &  -0.3949 &   1.3879\\
		&H  & -0.6329 &   1.2071  &  1.3059\\
		&H  &  0.0132  &  1.8001 &  -0.2631\\
	\end{tabular}
	\caption{The atomic coordinates (in \AA) of the molecules obtained from geometric optimization.}
	\label{table:geometry_one}
\end{table}

\begin{table}[h]
	\renewcommand{\arraystretch}{1.15}
	\begin{tabular}{c c | c | c | c}
		Molecule /& Atom & x & y & z \\
		\hline
		\textbf{LFEC} & F   & 2.0075   & 0.4516&   -0.2684\\
		&O&    0.4360 &  -1.2149&   -0.2662\\
		&O &  -1.2967  &  0.2398 &  -0.3826\\
		&O&   -1.4857&   -1.8199 &  -1.3620\\
		&C  &  0.9015  & -0.0702 &   0.4098\\
		&C &  -0.2574  &  0.9271  &  0.3467\\
		&C  & -0.7909  & -0.9376 &  -0.8490\\
		&H  &  1.2542 &  -0.3672  &  1.4079\\
		&H  & -0.6501  &  1.2159    &1.3298\\
		&H   & 0.0562  &  1.8125  & -0.2247\\
		&Li  &-2.3950 &  -3.0545 &  -2.2039\\
		\hline
		\textbf{PF$_6^-$} & P  &  0.0000   & 0.0000  &  0.0000\\
		&F   & 1.6258  &  0.0000   & 0.0000\\
		&F &  -1.6258   & 0.0000 &   0.0000\\
		&F  &  0.0000  &  1.6258  &  0.0000\\
		&F  &  0.0000  & -1.6258  &  0.0000\\
		&F  &  0.0000  &  0.0000  &  1.6258\\
		&F  &  0.0000  &  0.0000  & -1.6258\\
		\hline
		\textbf{LiPF$_6$} & P  & -2.9718  &  0.0296  & -0.0768\\
		&F &  -2.7478  &  1.4716  & -0.8105 \\
		&F  & -2.7395  & -0.5287  & -1.5942\\
		&F  & -4.0257 &  -1.1980  &  0.0996\\
		&F &  -1.7685 &  -1.2847  &  0.3039\\
		&F &  -1.7778   & 0.7455  &  1.0998\\
		&F  & -4.0348  &  0.8028  &  0.8833\\
		&Li &  -0.5410&   -0.4783  &  1.2492\\
	\end{tabular}
	\caption{The atomic coordinates (in \AA) of the molecules obtained from geometric optimization.}
	\label{table:geometry_two}
\end{table}

\begin{figure}[h]
  \includegraphics[width=0.9\columnwidth]{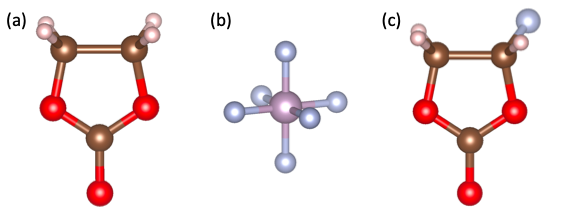}
  \caption{Illustration of three molecules studied in this study – (a) EC, (b) PF$_6^-$, (c) FEC. Red, brown, pink, violet, and blue spheres correspond to oxygen, carbon, hydrogen, phosphorus, and fluorine atoms, respectively. \label{fig:molecules}}
\end{figure}

\subsection{Molecular Hamiltonian}
\label{sec:molecular_hamiltonian}
Once the geometric optimization is complete, we obtain a fermionic Hamiltonian in second-quantized form, which can be written as:
\begin{equation}
	H = \sum_{i,j,k,l,\sigma,\sigma'} h_{ijkl}
	a_{i\sigma}^{\dagger}a_{j\sigma}a_{k\sigma'}^{\dagger}a_{l\sigma'}
	+ \sum_{i,j,\sigma} t_{ij} a_{i\sigma}^{\dagger}a_{j\sigma},
	\label{eq:fermionic_hamiltonian}
\end{equation}
where $h_{ijkl}$ and $t_{ij}$ are real numbers, $\sigma,\sigma'\in\{\uparrow,\downarrow\}$ are spins, and $i, j, k,$ and $l$ are the indices of the molecular orbitals. Here $a_{i\sigma}^{\dagger}$ and $a_{i\sigma}$ are fermion creation/annihilation operators, obeying the canonical anticommutation relation:
\begin{equation}
	\{a_{i\sigma}^{\dagger}, a_{j\sigma'} \} = \delta_{ij}\delta_{\sigma\sigma'}.
\end{equation}
Note that molecular orbitals are expressed by the basis set. In this study, STO-3G, DZ, 6-311G, cc-pVDZ, and cc-pVTZ basis sets are considered to represent various accuracies. We also assume the simulations are carried out in an open-shell scheme if the molecular system to be simulated contains unpaired electrons.

In Ref.~\cite{DoubleFactorized_MSFT}, the two-electron tensor $h_{ijkl}$ over $N$ orbitals is factorized by performing two levels of eigendecompositions, as follows:
\begin{equation} \label{eq:eigendecomp_of_h}
  h_{ijkl} = \sum_r L_{ij}^{(r)} L_{kl}^{(r)},
\end{equation}
where
\begin{equation} \label{eq:eigendecomp_of_L}
  L^{(r)} = \sum_m \lambda_m^{(r)} \overrightarrow{R_m^{(r)}} \cdot \left(\overrightarrow{R_{m}^{(r)}} \right)^T.
\end{equation}
Eq. \eqref{eq:eigendecomp_of_h} demonstrates a rank-$R$ factorization of $h_{ijkl}$ into two $N \times N$ symmetric matrices $L^{(r)}$ computed via eigendecomposition (where the eigenvalues are absorbed into the matrices), and Eq. \eqref{eq:eigendecomp_of_L} demonstrates a further rank-$M$ factorization of each $L^{(r)}$ matrix into vectors $\overrightarrow{R_m^{(r)}}$ and eigenvalues $\lambda_m^{(r)}$ via eigendecomposition. Here the summation ranges for $r \in [R]$ and $m \in [M]$ can be truncated by removing terms with small norms while keeping the total error smaller than the chemical accuracy. The result of the calculation is summarized in Table \ref{table:double_factorization}. We note that the values for $R$, $M$, and the one-norm $\alpha$ calculated for each molecule and basis set considered in this article are larger than those provided in Refs.~\cite{DoubleFactorized_MSFT,Lee2020} for FeMoco.

\begin{table*}[t]
\renewcommand{\arraystretch}{1.16}
  \begin{tabular}{c c |c|c|c|c|c}
    & & $N$ & $R$ & $M$ & $\alpha$ & $\max_r[M^{(r)}]$ \\
    \hline
    STO-3G & EC & 34 & 176 & 4493 & $5.29\times 10^2$ & 34\\
    & LEC & 39 & 211 & 6075 & $5.68\times 10^2$ & 39 \\
    & LREC & 38 & 210 & 6337 & $5.60\times 10^2$ & 38 \\
    & FEC & 38 & 198 & 5886 & $6.89 \times 10^2$ & 38 \\
    & LFEC & 43 & 231 & 7640 & $7.39 \times 10^2$ & 43 \\
    & PF$_6^-$ & 39 & 171 & 4352 & $1.19\times 10^3$ & 39 \\
    & LiPF$_6$ & 44 & 215 & 6964 & $1.27\times 10^3$ & 44 \\
    \hline
    DZ & EC & 68 & 540 & 29369 & $1.54\times 10^3$ & 68\\
    & LEC & 78 & 626 & 39538 & $1.77\times 10^3$ & 78 \\
    & LREC & 76 & 623 & 38452 & $1.77\times 10^3$ & 76 \\
    & FEC & 76 & 605 & 37345 & $1.94 \times 10^3$ & 76 \\
    & LFEC & 86 & 694 & 48626 & $2.21 \times 10^3$ & 86 \\
    & PF$_6^-$ & 78 & 572 & 26649 & $2.64\times 10^3$ & 78 \\
    & LiPF$_6$ & 88 & 689 & 45838 & $2.99\times 10^3$ & 88 \\
    \hline
    6-311G & EC & 90 & 768 & 56708 & $3.15\times 10^3$ & 90\\
    & LEC & 103 & 887 & 74916 & $3.65\times 10^3$ & 103 \\
    & LREC & 100 & 871 & 72017 & $3.60\times 10^3$ & 100 \\
    & FEC & 100 & 845 & 70235 & $3.95 \times 10^3$ & 100 \\
    & LFEC & 113 & 968 & 90140 & $4.50 \times 10^3$ & 113 \\
    & PF$_6^-$ & 99 & 745 & 45889 & $4.80\times 10^3$ & 99 \\
    & LiPF$_6$ & 112 & 889 & 76751 & $5.45\times 10^3$ & 112 \\
    \hline
   cc-pVDZ & EC & 104 & 959 & 83523 & $3.94\times 10^3$ & 104\\
    & LEC & 118 & 1105 & 109293 & $4.51\times 10^3$ & 118 \\
    & LREC & 113 & 1064 & 101529 & $4.29\times 10^3$ & 113 \\
    & FEC & 113 & 1040 & 99000 & $4.66 \times 10^3$ & 113 \\
    & LFEC & 127 & 1188 & 126748 & $5.29 \times 10^3$ & 127 \\
    & PF$_6^-$ & 102 & 809 & 58093 & $3.97\times 10^3$ & 102 \\
    & LiPF$_6$ & 116 & 978 & 91773 & $4.65\times 10^3$ & 116 \\
   \hline
      cc-pVTZ & EC & 236 & 2454 & 486578 & $2.70\times 10^4$ & 236\\
    & LEC & 266 & 2770 & 619960 & $3.08\times 10^4$ & 266 \\
    & LREC & 252 & 2631 & 557492 & $2.87\times 10^4$ & 252 \\
    & FEC & 252 & 2618 & 554546 & $3.12 \times 10^4$ & 252 \\
    & LFEC & 282 & 2936 & 696150 & $3.54 \times 10^4$ & 282 \\
    & PF$_6^-$ & 214 & 2009 & 315062 & $2.41\times 10^4$ & 214 \\
    & LiPF$_6$ & 244 & 2369 & 476246 & $2.84\times 10^4$ & 244 \\
    \hline
  \end{tabular}
  \caption{The result of the double factorization. Truncation error is set to 1mH. Here $N$ is the number of orbitals in the basis set and $\alpha$ is the norm of the Hamiltonian after the truncation. For the definition of $R$ and $M$, see Section \ref{sec:molecular_hamiltonian}. \label{table:double_factorization}}
\end{table*}

\subsection{Quantum algorithm description}
\label{sec:algo_description}
The main goal of the quantum component of the computation is to estimate the ground state energy of Eq.~\eqref{eq:fermionic_hamiltonian} for each molecule and basis set. The dominant approach in the literature is to use quantum phase estimation (QPE), which, given an eigenstate of the Hamiltonian describing a system, $H$, and a unitary encoding of a function of $H$, produces the corresponding eigenenergy.

For the unitary input, we choose a modern method known as \textit{qubitization}~\cite{Qubitization2019}. The dominant part of the quantum computation lies in the qubitization subroutine, and as such, optimizing this part of the computation is crucial (see Appendix \ref{sec:block_encoding_strategy}). Recently, methods to utilize the sparsity structure of the tensor $h_{ijkl}$ have shown great promise~\cite{Berry2019,DoubleFactorized_MSFT,Lee2020}. These developments have led to several orders of magnitude reduction in the computational cost for studying the FeMoco Hamiltonian, the oft-cited ``poster-child'' high-impact application of quantum computation~\cite{Reiher2016}. Our work builds upon these important recent developments.

Specifically, we shall adopt an approach taken in Ref.~\cite{DoubleFactorized_MSFT} to estimate the computational cost of estimating the ground state energy of the molecules described above. We note that other methods, \emph{i.e.,} the ones described in Ref.~\cite{Berry2019,Lee2020} may lead to a lower overall computational cost. We leave those studies for future work. 

\subsubsection{Ground state approximation}
One may be concerned with the necessity to prepare an eigenstate of the Hamiltonian as input to QPE; indeed, in practice, for most molecules of interest it is infeasible to prepare an exact eigenstate of $H$, and so one must instead prepare a trial ansatz wavefunction with non-negligible overlap with the true eigenstate (and in particular, the ground state). Concerns over the viability of preparing such an ansatz efficiently are encapsulated by the ``orthogonality catastrophe", the observation that the overlap between the true ground state and some ansatz wavefunction decreases exponentially as the system size increases~\cite{kohn}. However, it has been shown that there are sophisticated classical methods for preparing trial wavefunctions with sufficient overlap, for states of up to $\mathcal{O}(100)$ orbitals using simple-to-prepare states such as the single Slater determinant obtained from the Hartree-Fock method (alternatively, methods for multi-determinant state preparation can be used)~\cite{state_prep}. It should be noted that the resource counts for the number of qubits used in the algorithm presented in this article (shown in Table \ref{table:double_tcount}) include not only the system qubits whose state represents the wavefunction of the physical system, but also a plethora of ancillae. This state preparation routine only needs to act on the small subset of qubits that represent the system, which is only the number of orbitals (only $\mathcal{O}(10)-\mathcal{O}(100)$ for all molecules and basis sets; see Table \ref{table:double_factorization}).

\begin{table}[h]
\renewcommand{\arraystretch}{1.2}
  \begin{tabular}{c c |c|c}
    & & Number of logical qubits & T-count \\
    \hline
    STO-3G & EC & $2685$ & $6.32\times10^{10}$ \\
    & LEC & $3048$ & $8.07\times10^{10}$ \\
    & LREC & $2976$ & $7.95\times10^{10}$ \\
    & FEC & $2981$ & $9.56\times10^{10}$ \\
    & LFEC & $3342$ & $1.20\times10^{11}$ \\
    & PF$_6^-$ & $3137$ & $1.57\times10^{11}$ \\
    & LiPF$_6$ & $3507$ & $2.07\times10^{11}$ \\
    \hline
    DZ & EC & $10462$ & $5.41\times10^{11}$ \\
    & LEC & $12277$ & $7.57\times10^{11}$ \\
    & LREC & $11657$ & $7.25\times10^{11}$ \\
    & FEC & $11968$ & $8.01\times10^{11}$ \\
    & LFEC & $13509$ & $1.08\times10^{12}$ \\
    & PF$_6^-$ & $12280$ & $1.00\times10^{12}$ \\
    & LiPF$_6$ & $13817$ & $1.45\times10^{12}$\\
    \hline
    6-311G & EC & $14492$ & $1.69\times10^{12}$\\
    & LEC & $16548$& $2.37\times10^{12}$\\
    & LREC & $16073$ & $2.26\times10^{12}$\\
    & FEC & $16073$ & $2.45\times10^{12}$\\
    & LFEC & $18129$ & $3.31\times10^{12}$\\
    & PF$_6^-$ & $15912$ & $2.54\times10^{12}$\\
    & LiPF$_6$ & $18423$ & $3.79\times10^{12}$ \\
    \hline
   cc-pVDZ & EC & $16698$ & $2.68\times10^{12}$ \\
    & LEC & $18918$ & $3.69\times10^{12}$ \\
    & LREC & $18130$ & $3.32\times10^{12}$ \\
    & FEC & $18129$ & $3.57\times10^{12}$ \\
    & LFEC & $20855$ & $4.87\times10^{12}$ \\
    & PF$_6^-$ & $16382$ & $2.32\times10^{12}$ \\
    & LiPF$_6$ & $18600$ & $3.49\times10^{12}$ \\
   \hline
      cc-pVTZ & EC & $81958$ & $6.25\times10^{13}$ \\
    & LEC & $92345$ & $8.37\times10^{13}$ \\
    & LREC & $87490$ & $7.26\times10^{13}$ \\
    & FEC & $87498$ & $7.87\times10^{13}$\\
    & LFEC & $100139$ & $1.06\times10^{14}$ \\
    & PF$_6^-$ & $74348$ & $4.55\times10^{13}$ \\
    & LiPF$_6$ & $84721$ & $6.65\times10^{13}$ \\
    \hline
  \end{tabular}
  \caption{T-count $n_T$ and qubit count $n_L$ for each molecule and basis set that minimizes the computational volume $n_T \times n_L$. Note that these counts assume the serial application of magic states and do \emph{not} make use of any of the parallelization techniques discussed in Section~\ref{sec:official_result}. \label{table:double_tcount}}
\end{table}

\subsection{Fault-tolerant quantum algorithm cost model}
\label{sec:ft_algo_cost_model}
The cost of a quantum computation depends not only on the algorithm used, but also on the gate set used in the quantum computer and the cost of each gate. We consider the Clifford + T gate set. For our purposes, it will suffice to quantify the cost of the algorithm in terms of the following three quantities: (i) T-count, $n_T$; (ii) T-depth, $D_T$; and (iii) the number of logical qubits, $n_L$. T-count is the number of T-gates used in the quantum circuit, T-depth is the number of layers of commuting T-gates, and $n_L$ is the number of logical qubits that appear in the circuit description of the algorithm. We treat the cost of Clifford gates to be effectively equal to zero. As detailed in Appendix \ref{sec:compilation}, after compiling the quantum circuit in a certain way, the requisite cost of the Clifford gates will be negligible.

The parameters used to determine the scaling of these quantities depend greatly on the Hamiltonian pre-processing method used. In our case using the double factorization procedure (detailed in \ref{sec:molecular_hamiltonian}), these quantities can be determined by four parameters. These are (i) the number of orbitals, $N$; (ii) the ranks $M$ and $R$ that appear in the truncation procedure, and (iii) the norm of the Hamiltonian after the truncation, $\alpha$. See Table \ref{table:double_factorization} for the calculated values of these parameters for the molecules introduced above.

\textbf{Computational volume} When considering the overall cost of an algorithm, it is useful to consider the overall \emph{computational volume} of an algorithm; that is, a proxy for the space required to encode the problem of interest, multiplied by the time it takes to execute the algorithm. The space roughly corresponds to the number of logical qubits $n_L$, whereas the runtime roughly corresponds to the number of time slices needed to perform the gates in the computation. We typically define the overall volume to be $V_n \coloneqq n_L n_T$. However, in circumstances where magic state factories occupy a small percentage of the overall footprint, thus allowing us to potentially increase the number of magic state factories on the quantum computer, we may instead consider $V_D \coloneqq n_L D_T$. Viewed in this way, $V_n$ sets an upper bound on the volume of the algorithm as it assumes the \emph{serial} application of T gates, and $V_D$ sets a lower bound on the volume of the algorithm assuming that one can perform multiple T gates in parallel.

\subsection{Fault-tolerant overheads}
\label{sec:methods_overhead}
Here we detail the method by which we convert the algorithmic-level costs into architecture and hardware-level costs. We take into account the overhead due to quantum error correction and fault-tolerance, and detail the metrics and primitives specific to our fault-tolerance approach. We utilize the recently-introduced \textit{fusion-based quantum computing} (FBQC) ~\cite{FBQCpaper}, a universal model of quantum computation particularly suited to photonic quantum computation (for a brief introduction, see Appendix \ref{sec:fusion}). For the computations detailed in this article, we use FBQC to implement the more familiar surface-code type topological quantum error correction, and take our fault-tolerant operations to be lattice surgery operations on patches of surface-code encoded qubits (as we note in Section \ref{sec:generality}, many of the methods discussed here and the results that follow are general in nature, and are not features native exclusively to the FBQC paradigm).

\textbf{Fault-tolerance overhead metrics} The primitive building blocks for a fusion-based quantum computer are \emph{resource-state generators} (RSGs), devices that generate resource states on a fixed clock rate $f_{RSG}$. For simplicity, we consider resource states made up of 6 qubits in a ring called the ``6-ring" fusion network, which has been shown to be able to perform surface-code type error correction~\cite{FBQCpaper}. The relevant metric for the size of the computation is the number of resource state modules $n_{RSG}$ required and the number of RSG cycles $n_{cycles}$ required to implement the quantum algorithm of interest. From $n_{cycles}$ we can determine the total algorithm time $t_{algo}$ by dividing by the clock rate $f_{RSG}$.

To estimate the FT overhead, the algorithmic parameters we require are the T-gate count $n_T$ and the (logical) qubit count $n_L$. We set a tolerable failure rate for the entire computation $\epsilon_{\text{total}}$ – this is the rate at which the computation will fail due to errors in the device. We assume here that $\epsilon_{\text{total}}=1\%$ is sufficient, as the computation can be rerun for increased certainty as required. For example, since the required logical error rate is well below $10^{-10}$, an order of magnitude change in the target error makes at most a $10\%$ change in the code distance, which does not affect the estimates by much. From this we determine the tolerable error rate per gate $\epsilon_{\text{gate}}$, given by $\epsilon_{\text{gate}} = \epsilon_{\text{total}}/(n_T n_L)$ – this is the max error per gate that we can tolerate in order to keep the overall computational error rate below $\epsilon_{\text{total}}$.

\textbf{Noise model and logical operations} Logical operations are implemented by specific fusion networks realizing space-time channels for fault-tolerant surface code operations. The fault-tolerant operations we consider are based on lattice surgery~\cite{horsman2012surface,litinski2019game}, noisy magic state preparation~\cite{lodyga2015simple,brown2020universal,Interleaving}, along with magic state distillation~\cite{bravyihaahmagic,haah2017magic,haah2018codes,litinski2019magic}. Each T-gate is realized by a lattice surgery operation on a specified number of target logical qubits along with a distilled encoded $|T\rangle = (|0\rangle + e^{i\pi/4}|1\rangle)/\sqrt{2}$ state. Each fusion-network must be large enough such that the logical error rate per gate is less than $\epsilon_{\text{gate}}$.

The resource states that are produced are in general noisy, with further errors occurring during their propagation and measurement. In addition, fusions between resource states are also subject to noise. One can consider the combined effect of all error processes as resulting in erroneous outcomes on the fusions (for more details, see Ref.~\cite{FBQCpaper}). In particular, we characterize the noise on each fusion measurement by two parameters: $p_P$ and $p_E$, known as the Pauli error rate and erasure error rate, respectively. We assume that for each measurement, there is a probability $p_E$ that the outcome is erased, a probability of $p_P(1-p_E)$ that the outcome is incorrect (\emph{i.e.} bit-flipped but not erased), and a probability of $(1-p_E)(1-p_P)$ that the measurement is correct. We refer to $p=(p_P, p_E)$ as the physical error rate. These parameters also determine the error rate that magic states can be initially prepared with. By post-selecting based on erasure and syndrome information, one can prepare the initial magic states with logical Pauli error rates comparable to the physical Pauli error rate, provided physical error rates are sufficiently low~\cite{li2015magic}. Such protocols require only modest overhead.

For a physical error rate $p$, a generic lattice surgery operation in our scheme has a logical error rate
\begin{equation}\label{eqLER}
\epsilon_{\text{gate}} = Ae^{-Bd}
\end{equation}
where $A$ and $B$ are parameters depending on the error rate $p$ that are estimated by numerical simulations, and $d$ is the code distance (the minimal number of elementary errors required to cause a logical fault). We refer to the exponential decay of $\epsilon_{\text{gate}}$ (as a function of $A$, $B$, and $d$) as the \emph{sub-threshold scaling}. Each lattice surgery operation requires $2d^2$ RSGs in footprint and takes $d$ clock cycles (using the doubly-encoded patch lattice surgery method of Ref.~\cite{litinski2019game}). The size of the fusion network is therefore determined by the tolerable gate error rate and numerically estimated parameters  $d=1/B \log(A/\epsilon_{\text{gate}})$.

\textbf{Distillation protocol} With the code distance $d$ fixed, we can estimate the overhead required to distill encoded $|T\rangle$ states. Distilling encoded $|T\rangle$ states is a highly optimized protocol that depends on the parameters $\epsilon_{\text{gate}}$, $A$, $B$, $d$, and $p$. It requires $n_{\text{distill}}$ additional RSGs to prepare and route an encoded $|T\rangle$ state every $d$ clock cycles with error rate less than $\epsilon_{\text{gate}}$. We consider distillation factories based on multiple levels of the $15$-to-$1$ protocol~\cite{haah2017magic,haah2017magic_low} and in particular, the FBQC analogue of the implementation in Sec.~$3$ of Ref.~\cite{litinski2019magic}. Each level of this protocol consists of a distillation block in which the $15$-to-$1$ distillation circuit takes place, and a connector block which receives the output magic states of lower level distillation blocks and performs the required lattice surgery operation. Output magic states from each level are routed to higher level connector blocks using interleaving fiber. The number of distillation blocks at each level is chosen such that the magic state production rate is not less than the magic state consumption rate of the next higher level. We utilize two such factories, with a total footprint denoted by $n_{\text{distill}}$.

\section{Results}
\label{sec:official_result}
Here we provide an analysis of the algorithmic and architectural techniques used in this paper. On the algorithmic side, first we show the scaling of this algorithm in terms of qubits $n_L$, gate count $n_T$, and depth $D_T$ with respect to the metrics introduced in \ref{sec:molecular_hamiltonian}. Given that the system sizes we consider are large enough such that the estimated footprint of magic state distillation factories is relatively small (as mentioned above and detailed later), next we summarize strategies for parallelizing the circuits. We then compare computational cost of the serialized and parallelized circuits to each other.

On the architectural side, we first describe the fault-tolerance overhead of these algorithms in the FBQC model, and provide concrete runtime and footprint estimates for these algorithms. Then we comment on the relative footprint of the magic state factory, which leads us to a discussion on a novel method that allows injection of multiple distilled magic states (while keeping the number of Clifford gates comparatively small), and estimate the potential reduction in runtime associated with these techniques (this is fully detailed in Appendix \ref{sec:constant_time}). 

\subsubsection{Generality of results}
\label{sec:generality}
Here we emphasize that (most of) the results we present are applicable to a wide variety of architectures and are \textit{not} specific to the FBQC model; in particular, the estimate of the relative size of the magic state factory to the data blocks is a result of how expensive it is to perform lattice surgery operations when consuming T-states sequentially on large numbers of qubits, and is \textit{not} a result of the FBQC model. This conclusion would be true in \textit{any} architecture that uses the $15$-to-$1$ distillation protocols we highlight. The only results that are specific to the FBQC model are the ability to use \textit{interleaving} to perform space-time tradeoffs, and the underlying physical implementation of these schemes (\textit{i.e.}, counting RSGs rather than the more familiar notion of ``physical" qubits, and the clockspeed that is specific to our photonic platform).

\subsection{QPE cost estimate: scaling}
\label{sec:cost_estimate}
Here we discuss the algorithmic resource estimates for performing quantum phase estimation (QPE) for the molecular systems considered in this article. As mentioned in Section \ref{sec:algo_description}, the goal of QPE is to produce an estimate of an eigenenergy given unitary access to a Hamiltonian and an eigenstate of said Hamiltonian. The overall number of T gates $n_T$ for QPE can be expressed as the cost of embedding the Hamiltonian into a unitary circuit $n_{T,Q}$, and the number of queries one must make to this circuit to approximate an eigenenergy up to error $\epsilon_P$. In general, to approximate an eigenenrgy up to precision $\epsilon_P$, one must query the unitary input to QPE $\mathcal{O}\big(\frac{1}{\epsilon_P}\big)$ times. We can express the asymptotic T gate count as:
\begin{equation} \label{eq:cost_function_count}
 n_T= \min_{\epsilon_Q + \epsilon_P \leq \epsilon} n_{T,Q}(\epsilon_Q) \frac{\alpha \pi P}{\epsilon_P},
\end{equation}
where $n_{T,Q}(\epsilon_Q)$ is the T-count contribution from the qubitization subroutine~\cite{Qubitization2019}  (\textit{i.e.}, the cost of embedding the Hamiltonian into a unitary circuit; see Appendix \ref{sec:block_encoding_strategy}), and $(\alpha \pi P)/\epsilon_P$ is the T-count contribution from phase estimation (or the number of times one must query the qubitized unitary to produce an estimate of the eigenenergy). Here the constant $P$ depends on the choice of phase estimation technique; we take this to be $1/2$, following the analyses presented in~\cite{Babbush2018} and~\cite{Heisenberg2020}. Notice here that the number of queries scales linearly with the norm of the Hamiltonian $\alpha$ (introduced in Section \ref{sec:molecular_hamiltonian}). This is because QPE estimates an eigenphase which lies between $-\pi$ and $\pi$, but the spectrum of the Hamiltonian may not lie in this range. When we block encode/qubitize the Hamiltonian (as discussed Appendix \ref{sec:block_encoding_strategy}), we rescale the Hamiltonian's eigenspectrum by a factor of $\alpha$, so we need to query the qubitized unitary a number of times proportional to this rescaling factor. Finally, $\epsilon_P$ is the precision chosen in phase estimation, and $\epsilon_Q$ is a catch-all for the error due to approximations in the qubitization procedure. To achieve the optimal T-count, the above expression must be minimized with respect to $\epsilon_P$ and $\epsilon_Q$, subject to an imposed constraint on the overall error budget $\epsilon$ to achieve chemical accuracy.

Similarly, the depth of the circuit $D_T$ can be expressed as:
\begin{equation} \label{eq:cost_function_depth}
  D_T = \min_{\epsilon_Q + \epsilon_P \leq \epsilon} D_{T,Q}(\epsilon_Q) \frac{\alpha \pi P}{\epsilon_P},
\end{equation}
where $D_{T,Q}$ is the T-depth of the qubitization subroutine, and $(\alpha \pi P)/\epsilon_P$ is once again the number of times we must query the qubitization subroutine.

Following the analyses and gate-counting arguments presented in Refs.~\cite{DoubleFactorized_MSFT,Lee2020}, these quantities, in the leading order, can be expressed as follows:

\begin{equation}
\begin{split}
n_{T,Q} & = \mathcal{O}\left(\frac{M}{\lambda}+ N \beta \lambda + N\beta + N \right) \\ & \quad + \mathcal{O}(\sqrt{R \log (M)}  \\ &\quad + \sqrt{M\log (1/\epsilon_Q)}), \\ \\
D_{T,Q} &= \mathcal{O}\left(\frac{M}{\lambda}+  \log(N)\beta \right)  \\& \quad + \mathcal{O}(\sqrt{R \log (M)} \\ &\quad + \sqrt{M\log (1/\epsilon_Q)}), \\ \\
n_L&= N\beta(1+\lambda) + 2N + \mathcal{O}(\log (N/\epsilon_Q)),
\end{split}
\end{equation}

for a tunable, integer parameter $\lambda$ and $\beta = \lceil 5.652 + \log_2 N \cdot \alpha/\epsilon_Q \rceil$. While the expressions for $n_{T,Q}$ and $D_{T, Q}$ are broadly similar, note the linear vs. logarithmic dependence on $N$ for these expressions, respectively. This difference in dependence is explained in the following subsection (\ref{sec:algorithmic_parallelization}).

\subsection{Algorithmic parallelization}
\label{sec:algorithmic_parallelization}
As mentioned in Section~\ref{sec:ft_algo_cost_model}, it is useful to consider the \textit{computational volume} of a computation, which is lower (upper) bounded by $V_D$ ($V_n$). For the qubitization algorithm described in this paper, there are a handful of techniques and augmentations to subroutines one can exploit to realize a $V_D$ that is significantly smaller than $V_n$. These techniques range from known, trivial techniques, to novel and/or non-trivial optimizations. We detail these techniques in Appendix \ref{sec:algo_parallelization}, and quote the relevant results here.

\textbf{Gap between $V_D$ and $V_n$} In order to get a sense of the magnitude of computational volume improvements possible, we consider the resource costs of three scenarios: (i) optimizing for $V_n$ and applying T gates serially; (ii) optimizing for $V_n$ but applying multiple T gates in parallel; and (iii) optimizing for $V_D$ and applying T gates in parallel. We provide the resource estimates for each molecule and basis set for each scenario in tables \ref{table:double_tcount}, \ref{table:tdepth_savings}, and \ref{table:tdepth_savings_depth_opt}. Here we quote the order of magnitude cost for the smallest and largest classically intractable instances considered: full configuration interaction (FCI) picture of $\text{PF}_6^{-}$ in the cc-pVDZ basis as the smallest instance, and FCI picture of LFEC in the cc-pVTZ basis as the largest.

\begin{table*}[t]
\renewcommand{\arraystretch}{1.15}
  \begin{tabular}{c c | c | c | c }
    & & T-count & T-depth & $\frac{\text{T-count}}{\text{T-depth}}$ \\
    \hline
    STO-3G & EC & $6.32\times10^{10}$ & $5.8\times10^{9}$ & $10.89$\\
    & LEC & $8.07\times10^{10}$ & $7.94\times10^{9}$ & $10.17$\\
    & LREC & $7.95\times10^{10}$ & $8.11\times10^{9}$ & $9.80$\\
    & FEC & $9.56\times10^{10}$  & $9.38\times10^{9}$ &$10.19$\\
    & LFEC & $1.20\times10^{11}$  & $1.24\times10^{10}$ &$9.66$\\
    & PF$_6^-$ &  $1.57\times10^{11}$ & $1.38\times10^{10}$ & $11.37$\\
    & LiPF$_6$ & $2.07\times10^{11}$  & $1.99\times10^{10}$ &$10.37$\\
    \hline
    DZ & EC & $5.41\times10^{11}$ & $4.72\times10^{10}$ &$11.45$ \\
    & LEC & $7.57\times10^{11}$ & $7.00\times10^{10}$ &$10.82$\\
    & LREC & $7.25\times10^{11}$ & $6.8\times10^{10}$ &$10.67$\\
    & FEC & $8.01\times10^{11}$ & $7.27\times10^{10}$ &$11.01$\\
    & LFEC &  $1.08\times10^{12}$ & $1.06\times10^{11}$ &$10.18$\\
    & PF$_6^-$ & $1.00\times10^{12}$ & $7.45\times10^{10}$ & $13.48$ \\
    & LiPF$_6$ & $1.45\times10^{12}$ & $1.35\times10^{11}$ &$10.68$\\
    \hline
    6-311G & EC & $1.69\times10^{12}$ & $1.74\times10^{11}$ & $9.73$\\
    & LEC &  $2.37\times10^{12}$ & $2.63\times10^{11}$ & $9.01$\\
    & LREC &  $2.26\times10^{12}$ & $2.49\times10^{11}$ & $9.06$\\
    & FEC &  $2.45\times10^{12}$ & $2.67\times10^{11}$ & $9.18$\\
    & LFEC &  $3.31\times10^{12}$ & $3.66\times10^{11}$ & $9.02$\\
    & PF$_6^-$ &  $2.54\times10^{12}$ & $2.18\times10^{11}$ &$11.65$\\
    & LiPF$_6$ & $3.79\times10^{12}$ & $4.01\times10^{11}$ & $9.44$\\
    \hline
   cc-pVDZ & EC &  $2.68\times10^{12}$ & $2.97\times10^{11}$ & $9.03$\\
    & LEC & $3.69\times10^{12}$ & $4.41\times10^{11}$ & $8.36$\\
    & LREC &  $3.32\times10^{12}$ & $3.92\times10^{11}$ & $8.48$\\
    & FEC & $3.57\times10^{12}$ & $4.15\times10^{11}$ & $8.59$\\
    & LFEC &  $4.87\times10^{12}$ & $5.97\times10^{11}$ & $8.15$\\
    & PF$_6^-$ & $2.32\times10^{12}$ & $2.24\times10^{11}$ & $10.31$\\
    & LiPF$_6$ & $3.49\times10^{12}$ & $3.85\times10^{11}$ & $9.06$\\
   \hline
      cc-pVTZ & EC &  $6.25\times10^{13}$ & $5.78\times10^{12}$ &$10.80$\\
    & LEC &  $8.37\times10^{13}$ & $8.28\times10^{12}$ & $10.11$\\
    & LREC &  $7.26\times10^{13}$ & $6.95\times10^{12}$ &$10.44$\\
    & FEC &  $7.87\times10^{13}$ & $7.52\times10^{12}$ &$10.48$\\
    & LFEC &  $1.06\times10^{14}$ & $1.07\times10^{13}$ &$9.89$\\
    & PF$_6^-$ &  $4.55\times10^{13}$ & $3.57\times10^{12}$ &$12.74$\\
    & LiPF$_6$ &  $6.65\times10^{13}$ & $5.95\times10^{12}$ &$11.17$\\
    \hline
  \end{tabular}
  \caption{T-count and T-depth comparison when optimizing $V_n$. The last column shows the potential savings offered by parallelizing the application of T gates (merely swapping out T-count for T-depth). All instances offer a potential savings over $8$x, with the largest instance reducing by over $13$x.  \label{table:tdepth_savings}}
\end{table*}

\begin{table*}[t]
\renewcommand{\arraystretch}{1.15}
  \begin{tabular}{c c | c | c | c | c }
    & & T-count & T-depth & $\frac{\text{T-count}}{\text{T-depth}}$ & $\frac{\text{original} V_n}{V_D}$ \\
    \hline
    STO-3G & EC & $1.34\times10^{11}$ & $5.04\times10^{9}$ & $26.67$ & $12.53$\\
    & LEC & $1.69\times10^{11}$ & $6.83\times10^{9}$ & $24.68$ & $11.82$\\
    & LREC & $1.64\times10^{11}$ & $6.99\times10^{9}$ & $23.39$ & $11.37$\\
    & FEC & $1.99\times10^{11}$ & $8.11\times10^{9}$ &$24.56$ & $11.78$\\
    & LFEC & $2.46\times10^{11}$  & $1.08\times10^{10}$ &$22.87$ & $11.17$\\
    & PF$_6^-$ &  $3.49\times10^{11}$ & $1.11\times10^{10}$ & $31.42$ & $14.13$\\
    & LiPF$_6$ & $4.36\times10^{11}$  & $1.71\times10^{10}$ &$25.44$ & $12.05$\\
    \hline
    DZ & EC & $1.46\times10^{12}$ & $4.02\times10^{10}$ &$36.21$ & $13.44$ \\
    & LEC & $2.00\times10^{12}$ & $6.06\times10^{10}$ &$33.01$ & $12.49$\\
    & LREC & $1.90\times10^{12}$ & $5.89\times10^{10}$ &$32.31$ & $12.33$\\
    & FEC & $2.13\times10^{12}$ & $6.3\times10^{10}$ &$33.78$ & $12.72$\\
    & LFEC &  $2.79\times10^{12}$ & $9.15\times10^{10}$ &$30.44$ & $11.76$\\
    & PF$_6^-$ & $2.87\times10^{12}$ & $6.33\times10^{10}$ & $45.42$ &  $15.87$\\
    & LiPF$_6$ & $3.82\times10^{12}$ & $1.17\times10^{11}$ &$32.62$ & $12.35$\\
    \hline
    6-311G & EC & $4.31\times10^{12}$ & $1.51\times10^{11}$ & $28.59$ & $11.24$\\
    & LEC &  $5.83\times10^{12}$ & $2.27\times10^{11}$ &$25.66$ & $10.42$\\
    & LREC &  $5.57\times10^{12}$ & $2.16\times10^{11}$ &$25.83$ & $10.47$\\
    & FEC &  $6.09\times10^{12}$ & $2.31\times10^{11}$ &$26.35$ & $10.61$\\
    & LFEC & $7.99\times10^{12}$ & $3.34\times10^{11}$ &$23.92$ & $9.89$\\
    & PF$_6^-$ &  $6.97\times10^{12}$ & $1.88\times10^{11}$ &$37.04$ & $13.50$\\
    & LiPF$_6$ & $9.59\times10^{12}$ & $3.47\times10^{11}$ & $27.61$ & $10.92$\\
    \hline
   cc-pVDZ & EC &  $6.45\times10^{12}$ & $2.72\times10^{11}$ & $23.73$ & $9.87$\\
    & LEC & $8.58\times10^{12}$ & $4.03\times10^{11}$ & $21.28$ & $9.14$\\
    & LREC &  $7.78\times10^{12}$ & $3.57\times10^{11}$ &$21.76$ & $9.29$\\
    & FEC & $8.41\times10^{12}$ & $3.79\times10^{11}$ &$22.19$ & $9.41$\\
    & LFEC &  $1.12\times10^{13}$ & $5.46\times10^{11}$ &$20.53$ & $8.91$\\
    & PF$_6^-$ &  $6.07\times10^{12}$ & $1.94\times10^{11}$ &$31.31$ & $11.94$\\
    & LiPF$_6$ &  $8.47\times10^{12}$ & $3.51\times10^{11}$ &$24.10$ & $9.93$\\
   \hline
      cc-pVTZ & EC &  $1.96\times10^{14}$ & $5.27\times10^{12}$ & $37.16$ & $11.84$\\
    & LEC &  $2.55\times10^{14}$ & $7.64\times10^{12}$ & $33.42$ & $10.96$\\
    & LREC &  $2.24\times10^{14}$ & $6.4\times10^{12}$ & $34.99$ & $11.33$\\
    & FEC &  $2.44\times10^{14}$ & $6.93\times10^{12}$ & $35.14$ & $11.36$\\
    & LFEC &  $3.2\times10^{14}$ & $9.85\times10^{12}$ & $31.32$ & $10.72$\\
    & PF$_6^-$ & $1.54\times10^{14}$ & $3.08\times10^{12}$ & $49.99$ & $14.77$\\
    & LiPF$_6$ &  $2.12\times10^{14}$ & $5.43\times10^{12}$ & $39.03$ & $12.25$\\
    \hline
  \end{tabular}
  \caption{T-count and T-depth comparison when optimizing $V_D$. The second-to-last column shows the potential savings offered by parallelizing the application of T gates (merely swapping out T-count for T-depth). The final column shows the overall volume reduction when optimizing for $V_D$ vs. optimizing for $V_n$. All instances offer a potential savings over $8$x, with the largest instance reducing by over $15$x. These savings are greater than the savings demonstrated in Table~\ref{table:tdepth_savings}. \label{table:tdepth_savings_depth_opt}}
\end{table*}

\begin{itemize}
  \item[] \textbf{Optimizing for $V_n$ and applying T gates serially} Computational volume range: $\mathcal{O}(10^{16})$ to $\mathcal{O}(10^{19})$.
  \item[] \textbf{Optimizing for $V_n$ and applying T gates in parallel} Computational volume range: $\mathcal{O}(10^{15})$ to $\mathcal{O}(10^{18})$.
  \item[] \textbf{Optimizing for $V_D$ and applying T gates in parallel} Computational volume range: $\mathcal{O}(10^{15})$ to $\mathcal{O}(10^{18})$.
\end{itemize}

Across all instances, applying T gates in parallel (either optimizing for $V_n$ or optimizing for $V_D$) allows for an average order of magnitude reduction in resource requirements (given the ability to distill multiple magic states in parallel). See tables \ref{table:tdepth_savings} and \ref{table:tdepth_savings_depth_opt}.

\subsection{Fault-tolerant overheads}
\label{sec:overhead}

We now estimate the resources required to implement the above algorithms in a photonic fusion-based quantum computing (FBQC) architecture~\cite{FBQCpaper,Interleaving} based on surface codes~\cite{Kitaev2003,bravyi1998quantum,Dennis2002,Fowler2012}.  The overhead estimates in this section are based on performing T-gates \emph{sequentially}, one at a time. We will discuss advantages of performing T-gates in parallel in the next subsection. 

The total number of RSGs in the computation is just the sum of the number of RSGs necessary for each block of the overall architectural layout depicted in Fig.~\ref{fig:arch_schematic}; \textit{i.e.}, the data / ancilla block where we perform lattice surgery operations on $d \times d$ surface code patches, and the magic state distillation (MSD) block where we distill and inject $\ket{T}$ states. For the data / ancilla block we require $2d^2n_L$ RSGs, and we denote the footprint of the MSDs as $n_{\text{distill}}$ for a total of $n_{RSG} = 2d^2n_L + n_{\text{distill}}$ RSGs. The number of cycles required for the total algorithm is the product of the number of gates $n_T$ and the code distance $d$.

\begin{figure*}[t]
	\includegraphics[width=0.9\textwidth]{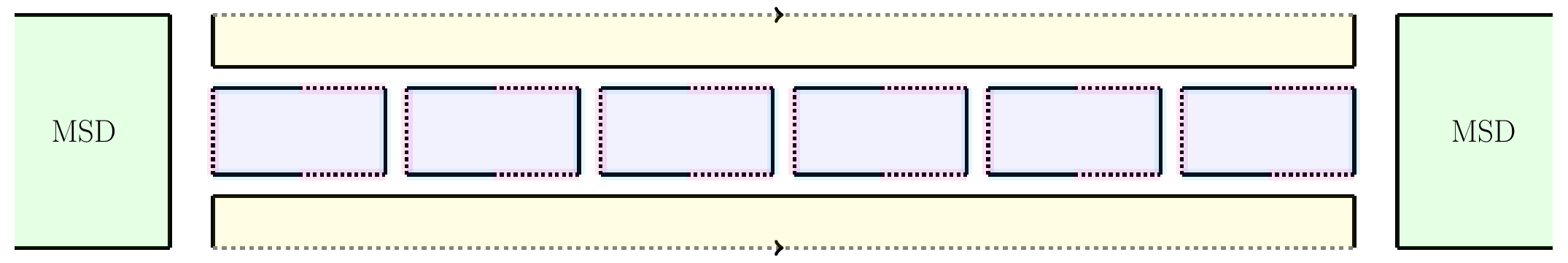}
	\caption{A schematic of the architecture used to produce the fault-tolerant estimates in Section~\ref{sec:overhead}. Resource state generators are laid out in a 2D plane with different regions responsible for different tasks. The central rectangles correspond to regions where logical qubits are produced from -- each patch produces two logical qubits propagating in time. The regions above and below the logical patches produce ancillary qubits, responsible for facilitating logic via lattice surgery; periodic boundary conditions are used as indicated by arrows. Two magic state distillation (MSD) factories are used which collectively inject distilled $|T\rangle$ states at the logical clock rate.}\label{fig:arch_schematic}
\end{figure*}

\textbf{Input parameters} We perform estimates for the fault-tolerant overheads of both footprint $n_{RSG}$ and time taken $n_{cycles}/f_{RSG}$ for a given operating point (given by a physical error rate $p$). As noted in Ref.~\cite{Interleaving}, all of the relevant components for resource state generation operate on GHz timescales -- for instance sources~\cite{paesani2019generation}, electronics and electro-optical modulators~\cite{wang2018integrated, eltes2020integrated} -- the relevant clockspeed for resource state generation is $f_{RSG} = 1$\;GHz. For many photonic architectures, loss is a dominant form of error. For photonic encodings such as the ‘dual rail’ encoding, loss is heralded and in addition to fusion ‘failure’ leads to erasures of the fusion outcomes. As such, we consider a simple noise model on the fusion outcomes that is dominated by erasures. See Ref.~\cite{FBQCpaper} for more details. We consider a ray in the noise parameter space defined by $(p_P(x), p_E(x)) = (x/10, x)$, $x\in[0,1]$. The threshold for the 6-ring fusion network along this ray occurs at $x^* = (4.71\pm 0.02)\times10^{-2}$, such that $(p_P(x^*), p_E(x^*)) = (4.71\times10^{-3},4.71\times10^{-2})$.  Here, the threshold is the maximal tolerable error rate (for this ray), below which, information can be protected arbitrarily well in the limit of large network size (see for example~\cite{Dennis2002} for thresholds in the context of toric codes and Ref.~\cite{FBQCpaper} for thresholds in the context of FBQC). We give estimates for two regimes:

\begin{itemize}
	\item[i)] High physical error rate: $x = 0.2\times x^*$ meaning $(p_P, p_E) = (9.4\times10^{-4}, 9.4\times10^{-3})$ where the logical error rate satisfies $A=0.4, B=1.1$.
	\item[ii)] Moderate physical error rate:$x = 0.1\times x^*$ meaning $(p_P, p_E) = (4.7\times10^{-4}, 4.7\times10^{-3})$ where the logical error rate satisfies $A=0.5, B=1.6$.
\end{itemize}
In the above, the parameters $A$ and $B$ defined by Eq.~(\eqref{eqLER}) are obtained by numerical simulations using the union-find decoder~\cite{delfosse2017almost} on the bulk 6-ring fusion network.

We remark that under this noise model and decoder, the $6$-ring fusion network has marginal thresholds of $p_P^* = 0.94\%$ and $p_E^* = 12\%$~\cite{FBQCpaper}, where the Pauli threshold can be increased using other decoders such as minimum weight matching~\cite{Dennis2002,kolmogorov2009blossom} (at the cost of increased run-time). We assume an algorithmic target error rate of $\epsilon_{\text{total}}=10^{-2}$. We also assume magic states are initially prepared with a logical Pauli error rate of $0.1\%$.

\textbf{Algorithmic footprint and runtime} We plot footprint and time estimates in Fig. \ref{fig:MSD_footprint_data} and note that the specific numbers can be found in Table~\ref{table:ft_overheads}. In the above high error rate (moderate error rate) estimates, code distances required for the data logical qubits range from $d=34$ to $d = 44$ ($d=24$ to $d=31$), meaning with trivial interleaving, between $1156$ to $1936$ ($576$ to $961$) resource state generators are required to produce an idling logical qubit. Additionally, with trivial interleaving in the high error rate (moderate error rate) regime, the number of RSGs required for the full computation is between $n_{RSG} = 6.4 \times 10^{6}$ and $n_{RSG} = 3.9 \times 10^{8}$  ($n_{RSG} = 3.2 \times 10^{6}$ and $n_{RSG} = 2.0 \times 10^{8}$), while the time taken in hours is between $t_{algo} = 0.60$ and $t_{algo}=1.3\times10^{3}$ ($t_{algo} = 0.42$ and $t_{algo}=9.1\times10^{2}$).

\begin{table*}[t]
	\renewcommand{\arraystretch}{1.15}
	\begin{tabular}{c c |c|c|c|c}
		& & $n_{RSG}$ [high error rate] & $t_{algo}$ (hours) [high error rate] & $n_{RSG}$ [moderate error rate]& $t_{algo}$ (hours) [moderate error rate] \\
		EC & STO-3G & $6.43\times 10^6$ & ~~$5.97\times 10^{-1}$ & $3.23\times 10^6$ & ~~$4.21\times 10^{-1}$ \\
		& DZ & $2.92\times 10^7$ & $5.56\times 10^0$ & $1.45\times 10^7$ & $3.91\times 10^0$ \\
		& 6-311G & $4.48\times 10^7$ & $1.84\times 10^1$ & $2.16\times 10^7$ & $1.27\times 10^1$ \\
		& cc-pVDZ & $5.16\times 10^7$ & $2.91\times 10^1$ & $2.49\times 10^7$ & $2.01\times 10^1$ \\
		& cc-pVTZ & $3.21\times 10^8$ & $7.63\times 10^2$ & $1.50\times 10^8$ & $5.20\times 10^2$ \\
		\hline
		LEC & STO-3G & $7.70\times 10^6$ & ~~$7.85\times 10^{-1}$ & $3.66\times 10^6$ & ~~$5.38\times 10^{-1}$\\
		& DZ & $3.61\times 10^7$ & $7.99\times 10^0$ & $1.70\times 10^7$ & $5.47\times 10^0$ \\
		& 6-311G & $5.11\times 10^7$ & $2.56\times 10^1$ & $2.46\times 10^7$ & $1.77\times 10^1$ \\
		& cc-pVDZ & $6.14\times 10^7$ & $4.10\times 10^1$ & $3.03\times 10^7$ & $2.87\times 10^1$ \\
		& cc-pVTZ & $3.62\times 10^8$ & $1.02\times 10^3$ & $1.80\times 10^8$ & $7.21\times 10^2$ \\
		\hline
		LREC & STO-3G & $7.52\times 10^6$ & ~~$7.73\times 10^{-1}$ & $3.57\times 10^6$ & ~~$5.30\times 10^{-1}$ \\
		& DZ & $3.43\times 10^7$ & $7.66\times 10^0$ & $1.61\times 10^7$ & $5.24\times 10^0$ \\
		& 6-311G & $4.97\times 10^7$ & $2.45\times 10^1$ & $2.39\times 10^7$ & $1.69\times 10^1$ \\
		& cc-pVDZ & $5.89\times 10^7$ & $3.69\times 10^1$ & $2.90\times 10^7$ & $2.58\times 10^1$ \\
		& cc-pVTZ & $3.43\times 10^8$ & $8.87\times 10^2$ & $1.60\times 10^8$ & $6.05\times 10^2$ \\
		\hline
		FEC & STO-3G & $7.56\times 10^6$ & ~~$9.29\times 10^{-1}$ & $3.58\times 10^6$ & ~~$6.37\times 10^{-1}$ \\
		& DZ & $3.52\times 10^7$ & $8.45\times 10^0$ & $1.66\times 10^7$ & $5.78\times 10^0$ \\
		& 6-311G & $4.97\times 10^7$ & $2.66\times 10^1$ & $2.39\times 10^7$ & $1.84\times 10^1$ \\
		& cc-pVDZ & $5.89\times 10^7$ & $3.96\times 10^1$ & $2.90\times 10^7$ & $2.77\times 10^1$ \\
		& cc-pVTZ & $3.43\times 10^8$ & $9.62\times 10^2$ & $1.60\times 10^8$ & $6.56\times 10^2$ \\
		\hline
		LFEC & STO-3G & $8.46\times 10^6$ & $1.17\times 10^0$ & $4.00\times 10^6$ & ~~$8.01\times 10^{-1}$ \\
		& DZ & $3.97\times 10^7$ & $1.14\times 10^1$ & $2.01\times 10^7$ & $8.07\times 10^0$ \\
		& 6-311G & $5.89\times 10^7$ & $3.67\times 10^1$ & $2.90\times 10^7$ & $2.57\times 10^1$ \\
		& cc-pVDZ & $6.77\times 10^7$ & $5.41\times 10^1$ & $3.34\times 10^7$ & $3.79\times 10^1$ \\
		& cc-pVTZ & $3.92\times 10^8$ & $1.29\times 10^3$ & $1.96\times 10^8$ & $9.09\times 10^2$ \\
		\hline
		$\text{PF}_6^-$ & STO-3G & $7.95\times 10^6$ & $1.52\times 10^0$ & $4.07\times 10^6$ & $1.09\times 10^0$ \\
		& DZ & $3.61\times 10^7$ & $1.06\times 10^1$ & $1.83\times 10^7$ & $7.53\times 10^0$ \\
		& 6-311G & $4.92\times 10^7$ & $2.75\times 10^1$ & $2.37\times 10^7$ & $1.90\times 10^1$ \\
		& cc-pVDZ & $5.06\times 10^7$ & $2.51\times 10^1$ & $2.44\times 10^7$ & $1.74\times 10^1$ \\
		& cc-pVTZ & $2.78\times 10^8$ & $5.43\times 10^2$ & $1.36\times 10^8$ & $3.79\times 10^2$ \\
		\hline
		$\text{LiPF}_6$ & STO-3G & $9.35\times 10^6$ & $2.07\times 10^0$ & $4.54\times 10^6$ & $1.43\times 10^0$ \\
		& DZ & $4.06\times 10^7$ & $1.53\times 10^1$ & $2.06\times 10^7$ & $1.08\times 10^1$ \\
		& 6-311G & $5.98\times 10^7$ & $4.21\times 10^1$ & $2.95\times 10^7$ & $2.95\times 10^1$ \\
		& cc-pVDZ & $6.04\times 10^7$ & $3.88\times 10^1$ & $2.98\times 10^7$ & $2.71\times 10^1$ \\
		& cc-pVTZ & $3.32\times 10^8$ & $8.13\times 10^2$ & $1.55\times 10^8$ & $5.54\times 10^2$ \\
		\hline
	\end{tabular}
	\caption{Fault-tolerant space and time overheads assuming an interleaving ratio of $L_{\text{intl}}=1$ -- this provides the quickest computation at the expense of the largest footprint. The high and moderate error rate regimes are defined in Section~\ref{sec:overhead}.}
	\label{table:ft_overheads}
\end{table*}

\textbf{Interleaving trade-offs} Finally, we remark on the possible reductions in footprint due to interleaving, a feature specific to FBQC. As we have mentioned, interleaving ratios of up to several thousand are certainly possible using low-loss optical fiber. One can readily reduce the number of RSGs required by a factor $\sim 5\times 10^3$ while increasing the time taken by the same factor, compared to the estimates performed with trivial interleaving. In this regime, a single RSG is sufficient to generate several logical qubits. This brings the required footprint for the present computations to between $\sim10^3$ and $\sim10^5$ RSGs, even in the case of high error rate. We remark that the additional fiber loss at these interleaving ratios has very little impact on the photon loss threshold~\cite{Interleaving}. Thus we have neglected any potential changes to the sub-threshold scaling for these interleaving ratios, as we expect them to be small.

\subsection{Relative footprint of magic state factories}
We remark (perhaps surprisingly) that the footprint utilized for magic state distillation in the above computations is no more than $2\%$ of the total computational footprint. For the more costly basis sets of cc-pVDZ and cc-pVTZ, the magic state distillation utilizes less than $0.3\%$ of the footprint. Indeed, in the above estimates (which use only two levels of distillation -- the first level of which is relatively small) the footprint required to distill and route a magic state every $d$ clock cycles requires the equivalent RSG footprint of between $78$ and $120$ full size logical qubits (\emph{i.e.} $78d^2$ to $120d^2$ RSGs). As the qubitization approach for the above molecules requires many thousands of logical qubits, the distillation overhead to produce one $|T\rangle$ state per logical clock cycle is comparatively small. However for other methods that use relatively fewer logical qubits (such as those based on Trotterization), the distillation cost becomes more significant. Moreover, one may utilize higher-rate distillation protocols and perform more T gates in parallel (where possible) to perform faster quantum computations (as we discuss in Appendix \ref{sec:constant_time}).

Here we again reiterate that the relative footprint of magic state factories is a consequence of how expensive lattice surgery operations are in a surface code computation over many logical qubits, and \textit{not} a consequence of the FBQC paradigm.

\subsection{Constant-time PPMs}
\label{sec:results_ppms}
The resource estimates above are based on the sub-threshold scaling behavior of the FBQC scheme in Ref.~\cite{FBQCpaper}, making use of one of the quantum computing architectures in Ref.~\cite{litinski2019game}. We propose an alternative architecture in this section, in the hope that it may reduce the overall space-time cost of the algorithm. As we shall later see, we have reasons to believe that the newly proposed architecture can reduce the computation time without incurring a significant amount of extra footprint. However, let us emphasize that, unlike the results above, we have \emph{not} carried out a Monte Carlo study on the threshold and sub-threshold behavior of this new scheme. A fair comparison between the two can be made only via a rigorous numerical study, which we leave for future work.

We previously observed that if we use a state-of-the-art quantum algorithm to simulate the molecules relevant to battery research, the magic state factory constitutes only up to $2\%$ of the entire quantum computer. Therefore, if magic state distillation is the main bottleneck of the quantum computation, one may be able to reduce the computation time by increasing the number of magic state factories, and injecting the distilled magic states appropriately.

In this section, we propose a different scheme which is based primarily on Litinski's scheme~\cite{litinski2019game,litinski2019magic}. Litinski's scheme uses Pauli Product Measurement (PPM) and preparation of states in the set $\{ |0\rangle, |+\rangle, |T\rangle\}$. PPM refers to a non-destructive measurement of the following observable:
\begin{equation}
  \bigotimes_{n=1}^{N} P_n, \label{eq:PPM}
\end{equation}
where $P_{n} \in \{I, X, Y, Z \}$ and $P_n$ is the logical Pauli operator of the $n$'th qubit, which ranges from $1$ to $N$. Preparation of $|0\rangle$ and $|+\rangle$ can be implemented in $\mathcal{O}(1)$ time. On the other hand, the time to prepare $|T\rangle$ depends on the choice and the number of magic state factories. Performing a PPM takes $\mathcal{O}(d)$ time. Therefore, if the number of magic state factories is abundant, the main bottleneck becomes the speed of the PPM. Using Fowler’s technique \cite{fowler2013timeoptimal}, one can speed up the computation. However, in this scheme, a $k$-fold increase in speed necessitates a $k$-fold increase in footprint.

In Appendix \ref{sec:constant_time}, we detail an $\mathcal{O}(1)$-time implementation of an arbitrary PPM which does not incur such additional footprint. Moreover, all the physical operations can be made local in two spatial dimensions. Therefore, optimistically speaking, in the regime in which the number of $|T\rangle$-state factories is abundant so that more than one $|T\rangle$ state can be distilled in a single logical clock cycle, one can expect to be able to inject them all in a single logical clock cycle, provided that there are no more than $\mathcal{O}(d)$ of them. First in \ref{sec:ppm_circuit} we discuss an abstract circuit model that implements a PPM; then in \ref{sec:fast_ppm} we explain how this protocol can be implemented on the surface code~\cite{bravyi1998quantum} and show that each step of this protocol can be executed in constant time; then in \ref{sec:compilation} we explain how this method can be used to compile a general quantum algorithm; and finally in \ref{sec:speedups} we discuss the expected speedup from this approach.

\textbf{Expected speedup using $\mathcal{O}(1)$-time PPMs} Being able to execute an arbitrary PPM in constant time importantly does \emph{not} mean one can execute an entire sequence of arbitrary PPMs in $\mathcal{O}(1)$ time. There are two important constraints that put a cap on the expected achievable runtime speedup: our compilation method is sensitive to whether or not successive PPMs in a sequence commute with one another, and the total number of PPMs one can execute in a single logical clock cycle must be smaller than the code distance $d$. These constraints are discussed in detail in Appendix \ref{sec:speedups}.

With these constraints in mind, the speedup we can achieve using this approach is limited by the ratio of T-count and T-depth. In our setup we can optimistically expect up to a $15$-fold speedup with a negligible change in the footprint of the device. However, note that whether this is possible or not depends on many microscopic details, such as the timescale needed for different physical operations. Moreover, we would like to emphasize that the threshold and the sub-threshold behavior of the code may change when we implement the transversal gates. The exact extent to which we can speed up the computation using our approach can only be determined by carefully inspecting all of these different factors. These studies are beyond the scope of this paper and are left for future work.

Keeping these caveats in mind, we can estimate the optimistic computation time for a single phase estimation for the molecules described to be on the order of $1\sim 4$ hours for the cc-pVDZ basis set and $35\sim 90$ hours for the cc-pVTZ basis set.

\section{Discussion}
\label{sec:discussion}

A great deal of progress on the fault tolerant resource estimates of Hamiltonian simulation algorithms has been driven by a focus on specific molecules, like FeMoco. Such studies have served as benchmarks for the introduction of new algorithmic techniques and improvements \cite{Reiher2016,Babbush2018,Berry2019,DoubleFactorized_MSFT,Lee2020}. Here we propose a new class of commercially-interesting benchmark molecules: the constituent chemicals of Li-ion battery electrolytes. We estimate the quantum computational cost of simulating these molecules, first in terms of architecture-agnostic parameters like T-count and qubit count, and then by compiling the logical operations into the primitives of a fusion-based quantum computing architecture~\cite{FBQCpaper}.

For the smallest instance that is classically intractable (full configuration interaction simulation of PF$_6^-$ using the cc-pVDZ basis at $1$\;mHartree precision), we estimated a total of $16,382$ logical qubits and $2.32\times10^{12}$ T gates (see Table~\ref{table:double_tcount}); after compilation, this corresponds to an estimated footprint of $\sim 24$ million RSGs and a runtime of less than $1$ day, in the moderate error rate case (see Table~\ref{table:ft_overheads}). Using the recently-introduced interleaving technique~\cite{Interleaving}, we can perform linear space-time trade-offs between runtime and footprint; \emph{e.g.} using an interleaving ratio of $24$ for the same aforementioned instance, we can instead have a footprint of $\sim 1$ million RSGs and a runtime of $\sim 2.5$ weeks.

Molecules like FeMoco live in a regime where the size of the magic state factory (MSF) is relatively large, so as to make infeasible distilling and injecting multiple distilled magic states in the same time step; by exploring larger instances of molecules of commercial interest, we challenge this intuition. In this algorithmic regime with $10^4 - 10^5$ logical qubits, the relative size of the MSF is small enough that opting to have multiple MSFs is no longer a prohibitive cost.

Understanding the cost of MSFs in this regime to be minimal, we explore the degree to which we can parallelize the distillation and consumption of magic states in the algorithm studied here. The drastic gap between the T-depth and T-count of various subroutines in the algorithm suggests on average an order of magnitude potential reduction in computational runtime (up to $\sim 15$x for the largest instance; see Table~\ref{table:tdepth_savings}). In order to exploit this potential parallelization, we need not only multiple MSFs, but also fast measurements so as to not be bottlenecked by the consumption of magic states. With this motivation in mind, we propose a novel method for performing constant time PPMs. We emphasize that the effect of the transversal gates used in this procedure on the threshold and sub-threshold scaling is \emph{not} known; we leave this study for future work.

We note the great deal of flexibility in space-time tradeoffs afforded us both algorithmically and architecturally. Algorithmically, similar magnitudes of computational volume reduction ($\sim 10$x) are possible in a variety of ways, mostly revolving around the choice of $\lambda$ in the ubiquitous $\textsc{QROM}$ subroutine (discussed in Appendix \ref{sec:lambda}). Na\"ively, these savings translate to runtime savings in the fully-compiled resource estimate; however, interleaving allows us to easily distribute these savings between runtime and footprint. If we could perform the constant time PPMs detailed above and exploit the $\sim 10$x reduction in computational volume, the resource estimate for the small instance mentioned above could go from a footprint of $\sim 1$ million RSGs and a runtime of $\sim 2.5$ weeks to $\sim 0.1$ million RSGs and $\sim 2.5$ weeks, $\sim 1$ million RSGs and $\sim 1.7$ days, or $\sim0.5$ million RSGs and $\sim 3.5$ days by distributing the savings fully into footprint, fully into runtime, or evenly between the two, respectively.

Future studies in this regime should explore alternative Hamiltonian simulation algorithms that may be amenable to massive parallelization, as well as further algorithmic improvements tailored to depth considerations. The computational costs of the problems studied here would decrease dramatically given the ability to distill and consume multiple magic states in the same time slice. Ultimately, the overall volume savings that can be had may benefit greatly from continued co-development between architectures and algorithms. The methods introduced here serve to motivate further research into algorithmic and architectural techniques in order to realize the potential gains to be had in reducing the computational volume of problems in quantum chemistry.

%\bibliography{bib}

\section{Acknowledgments}
The authors would like to thank Daniel Litinski for both helpful discussions on the method for fast magic state injection, as well as support calculating the relative footprint of magic state factories. We would also like to thank Sara Bartolucci, Patrick Birchall, Hector Bomb\'in, Hugo Cable, Axel Dahlberg, Chris Dawson, Andrew Doherty, Megan Durney, Mercedes Gimeno-Segovia, Nicholas Harrigan, Eric Johnston, Konrad Kieling, Kiran Mathew, Ryan Mishmash, Sam Morley-Short, Naomi Nickerson, Andrea Olivo, Mihir Pant, Fernando Pastawski, Terry Rudolph, Karthik Seetharam, Jake Smith, Chris Sparrow, Jordan Sullivan, Andrzej P\'erez Veitia, and all our colleagues at PsiQuantum for useful discussions. \textbf{Funding:} This work was funded by PsiQuantum and Mercedes-Benz Research and Development North America. \textbf{Author contributions:} EL performed the geometric optimization for each molecule and provided the background information on Li-ion battery electrolytes. YL performed the double factorization calculations for the Hamiltonians of each molecule and basis set. WP and SP performed the algorithmic resource estimates and parallelization analysis of the circuits. SP also developed the ``Gizens" rotation. SR performed the architecture-specific fault-tolerant resource estimates and calculated the relative footprint of the magic state factory for each molecule and basis set. IK developed the constant-time PPM technique and described the compilation methods for our fault-tolerant operations. \textbf{Competing interests:} The authors declare that they have no competing interests. \textbf{Data availability:} All data needed to evaluate the conclusions in the paper are present in the paper. Additional data related to this paper may be requested from the corresponding authors. Correspondence and requests for materials should be addressed to William Pol and Eunseok Lee.

\appendix

\section{Introduction to the fusion-based quantum computing scheme}
\label{sec:fusion}

Fusion-based quantum computation (FBQC) is a universal model of quantum computation particularly suited to photonic quantum computation (for more details, see Ref.~\cite{FBQCpaper}). FBQC shares similarities with the well-known measurement-based model of quantum computation (MBQC)~\cite{Raussendorf2003,raussendorf2006fault,raussendorf2007topological}, in the sense that the computation is carried out by a sequence of measurements and feed-forwards instead of unitary gates; however, there are also important differences. In MBQC, one often begins with a single resource state (\emph{i.e.} a cluster state) over a large number of qubits, and carries out the computation via a sequence of single-qubit measurements and feed-forward operations. In contrast, in FBQC one begins with a \textit{set} of small resource states, and the computation proceeds by applying certain one- and two-qubit measurements onto this set of resource states. In FBQC quantum information is stored and processed by repeatedly teleporting it between different resource states. The resource states required in FBQC are small, fixed entangled states involving a few qubits (constant in the size of the computation)  which are ``fused" together using one- and two-qubit measurements. Any (fault-tolerant) computation can be expressed as a sequence of such fusions (\emph{i.e.} measurements) between a number of resource states arranged in (2{+}1)D spacetime. This layout of resource states in spacetime and the measurements between them is known as a \emph{fusion network}.

For simplicity, we focus on the the $6$-ring fusion network introduced in Ref.~\cite{FBQCpaper}, where the resource states are $6$-qubit cluster states on a ring, and the necessary fusion measurements consist of bell basis measurements, \emph{i.e.,} the measurement of both two-qubit Pauli operators $X \otimes X$ and $Z \otimes Z$, along with single-qubit Pauli measurements and $(X + Y)/\sqrt{2}$ magic state basis measurements. The $6$-ring fusion network can be understood as implementing surface-code type topological quantum error correction.

\textbf{Resource state generators and interleaving} The primitive building block for a fusion-based quantum computer is a \emph{resource-state generator} (RSG). RSGs are arranged in arrays, each producing resource states at a fixed clock rate $f_{RSG}$, the qubits of which can be routed to different measurement devices to undergo single- and two-qubit measurements in a manner determined by the fusion network. For fusion networks based on the surface code -- such as the $6$-ring fusion network -- a $2$D array of RSGs is sufficient for fault-tolerant universal quantum computation. With access to a quantum memory, such as optical fiber, a single RSG can produce many simultaneously-existing resource states, which can allow for large quantum computations to be performed using much fewer RSGs than one might typically expect. Namely, with optical fiber as a quantum memory, one can utilize \textit{interleaving} to perform linear space-time tradeoffs in the fault-tolerant quantum computation~\cite{Interleaving}. Thus in effect, interleaving reduces the number of required RSGs by a factor $L_{\text{intl}}$, while increasing the time required to execute a fusion network by the same factor $L_{\text{intl}}$. We refer to the case with $L_{\text{intl}}=1$ as trivial interleaving.

Note that optical fiber offers exceptionally low transmission loss rates of $0.2$\;dB/km ($4.5 \%$/km) or less (at $1550$\;nm wavelengths)~\cite{Li2020}. It has been shown that for the $6$-ring fusion network, using fibers of length $1000$\;m or more is readily possible without introducing too much error or significantly degrading the threshold (indeed, only some of the photons in an interleaved $6$-ring fusion network must travel through this long fiber)~\cite{Interleaving}. For RSGs with \;gigahertz clock rates, roughly $5000$ photonic qubits can simultaneously be stored in a $1$\;km fiber~\cite{Interleaving}. Thus we presume interleaving ratios $L_{\text{intl}}$ of up to $5000$ to be achievable in practice.

\section{Summary of block encoding techniques and primitives for double-factorized Hamiltonians}
\label{sec:block_encoding_strategy}

Below we summarize the strategy for qubitizing the double-factorized form of the Hamiltonian. Quantum phase estimation (QPE) requires inputting a unitary to then extract eigenvalues. For a Hamiltonian of interest $H$, the input unitary is some function of the Hamiltonian; for example, $e^{-iHt}$. However, an attractive alternative that has led to the most efficient resource estimates for quantum chemistry simulations is ``qubitization”, where one can implement alternative functions of the Hamiltonian; up to errors due to rotation gate synthesis, some unitaries can be implemented exactly~\cite{Qubitization2019}. Performing qubitization relies heavily on ``block-encoding” $\mathcal{B}$, a process by which one encodes a Hamiltonian in a ``block'' of a unitary acting on a larger Hilbert space:
\begin{equation}
\mathcal{B}\big(H\big) =
\begin{bmatrix}
H & \dots \\
\vdots & \ddots
\end{bmatrix}.
\end{equation}
This allows us to embed a non-unitary matrix (such as a Hamiltonian) into a larger unitary circuit, which can now be input into QPE.

\textbf{Block-encoding primitives.} Block-encoding requires an input model to access the terms and coefficients of the Hamiltonian. A common input model is the ``linear combination of unitaries” (LCU) model, where one can express the Hamiltonian as a sum of easier-to-implement unitaries (typically Pauli matrices)~\cite{LCU2012}. In its simplest form for un-factorized Hamiltonians, this only requires two subroutines: one to load the Hamiltonian coefficients onto the ancillary register (often referred to as $\textsc{prepare}$ ~\cite{Babbush2018,Berry2019,Lee2020}), and another to selectively apply the corresponding Hamiltonian term (referred to as $\textsc{select}$~\cite{Babbush2018,Berry2019,Lee2020}) onto the system register.

\textbf{Block-encoding + LCU strategy}. The idea is that the ancilla register can be prepared in a superposition state over all possible indices for the terms in the Hamiltonian, weighted by the square root of the coefficient of each term (normalized by the norm $\alpha$), and then this same index ancilla register can now be used to apply the corresponding Hamiltonian term onto the system register, weighted by the appropriate amplitude. Finally, we uncompute the state preparation routine on the ancilla register; now, in the subspace where the ancillae are in the all-zero state, the Hamiltonian has been applied onto the system qubits.

\textbf{Data-loading oracle.} Another important and ubiquitous subroutine is the \textit{data-lookup oracle}, often referred to as $\text{QROM}$ ~\cite{Babbush2018,Berry2019}. Given a list of $K$ $b$-bit elements $\vec{a} = [a_0, ..., a_{K - 1}]$, $\text{QROM}$ performs the following task:
\begin{equation}
  \text{QROM} \ket{x} \ket{0} = \ket{x} \ket{a_x}
\end{equation}
One can use a number of ancillae to trade between space and time, resulting in space-time optimized costs for this unitary~\cite{Berry2019,LowKliuchnikov}. The asymptotic T-count of the $\text{QROM}$ in Ref.~\cite{Berry2019,LowKliuchnikov} is:
\begin{equation} \label{eq:qrom_cost}
  \left\lceil{\frac{K}{\lambda}}\right\rceil + b \cdot (\lambda - 1)
\end{equation}
where $\lambda$ is a tunable power-of-two number of copies for the $b$-bit approximation of the elements in $\vec{a}$. We highlight this subroutine for $3$ reasons: (i) it is used throughout the double-factorized qubitization circuit both as a standalone routine and as a subroutine for $\textsc{prepare}$; (ii) the plurality of the space-time cost of qubitizing the double-factorization comes from one large data-lookup (as summarized in Table~\ref{table:subroutine_split}); and (iii) because the choice of how to tune $\lambda$ greatly informs the overall cost of the algorithm either in terms of T-count or T-depth.

\begin{table}[h]
\renewcommand{\arraystretch}{1.1}
  \begin{tabular}{c c | c | c}
    & & \multicolumn{1}{p{3cm}|}{\centering Basis-change rotations \\ (\% total volume)} &  \multicolumn{1}{p{3cm}}{\centering Rotation data-lookup \\ (\% total volume)} \\
    \hline
    STO-3G & EC & $21.5$ & $32.6$ \\
    & LEC & $21.0$ & $35.3$ \\
    & LREC & $20.4$ & $36.1$ \\
    & FEC & $20.8$ & $35.0$ \\
    & LFEC & $20.3$ & $37.5$ \\
    & PF$_6^-$ & $23.3$ & $31.0$ \\
    & LiPF$_6$ & $21.5$ & $35.8$ \\
    \hline
    DZ & EC & $8.4$ & $51.5$ \\
    & LEC & $8.2$ & $53.5$ \\
    & LREC & $8.1$ & $53.6$ \\
    & FEC & $8.2$ & $53.1$ \\
    & LFEC & $7.9$ & $55.0$ \\
    & PF$_6^-$ & $9.2$ & $49.8$ \\
    & LiPF$_6$ & $8.2$ & $54.3$\\
    \hline
    6-311G & EC & $7.7$ & $56.1$\\
    & LEC & $7.4$ & $58.1$\\
    & LREC & $7.4$ & $57.9$\\
    & FEC & $7.5$ & $57.6$\\
    & LFEC & $7.1$ & $59.6$\\
    & PF$_6^-$ & $8.7$ & $53.3$\\
    & LiPF$_6$ & $7.7$ & $57.5$ \\
    \hline
   cc-pVDZ & EC & $7.1$ & $59.4$ \\
    & LEC & $6.7$ & $61.7$ \\
    & LREC & $6.8$ & $61.1$ \\
    & FEC & $6.8$ & $60.8$ \\
    & LFEC & $6.6$ & $62.6$ \\
    & PF$_6^-$ & $8.1$ & $55.5$ \\
    & LiPF$_6$ & $7.2$ & $59.7$ \\
   \hline
      cc-pVTZ & EC & $2.7$ & $71.0$ \\
    & LEC & $2.6$ & $72.2$ \\
    & LREC & $2.6$ & $71.7$ \\
    & FEC & $2.6$ & $71.7$ \\
    & LFEC & $2.5$ & $72.7$ \\
    & PF$_6^-$ & $3.0$ & $68.5$ \\
    & LiPF$_6$ & $2.7$ & $70.7$ \\
    \hline
  \end{tabular}
  \caption{Percentage of total qubitization volume per subroutine, where volume is taken to mean $n_T \times n_L$. The rotation data-lookup/$\text{QROM}$ occurs twice (so its volume is doubled in this table), and the rotations occur four times (so its volume is quadrupled in this table). As the basis set increases in size from $\text{STO-3G}$ to $\text{cc-pVTZ}$, the percentage share of the rotations decreases drastically, and the percentage share of the $\text{QROM}$s nearly doubles. Even still, together these subroutines dwarf the volume contributions of other subroutines. \label{table:subroutine_split}}
\end{table}

\textbf{Double-factorization}. The block-encoding must be modified in order to account for the double-factorization of the Hamiltonian. The high-level alterations follow from the following observation: by adopting the Majorana representation of the fermion operators, one can re-express the Hamiltonian as a sum of products of one-body terms in the following way (see Ref.~\cite{DoubleFactorized_MSFT}, and Ref.~\cite{Lee2020}, Appendix C):
\begin{equation}  \label{eq:ham_rewrite}
 \begin{aligned}
 H_{DF} = \big(\text{Offset}\big)\mathcal{I} + \text{One}_{L^{-1}} + \text{Two}_H, \\
 \text{Two}_H = \frac{1}{4} \sum_r \Lambda_r \lvert \lvert L_{r} \rvert \rvert _{SC}^2 T_2 \bigg[\frac{\text{One}_{L^r}}{\lvert \lvert L_{r} \rvert \rvert _{SC}}\bigg], \\
 \text{One}_L = \frac{1}{2} \sum_m \lambda_m \gamma_{\vec{R}_{m, 0}} \gamma_{\vec{R}_{m, 1}}
 \end{aligned}
 \end{equation}
where $T_2 (x)=2x^2-1$ is a Chebyshev polynomial, $\lvert \lvert L_{r} \rvert \rvert _{SC}$ is the one-norm of the eigenvalues of each matrix $L^r$, the coefficients $\Lambda_r$ result from normalizing each $L^r$, and $\gamma_{\vec{R_m}}$ is a basis-rotated Majorana operator. The one-body term is expressed as a sum of products of basis-rotated Majorana operators, and the two-body term is expressed as the sum of polynomials of one-body terms. Also, $H_{DF}$ contains an identity offset that can effectively be ignored as it contributes a constant shift in energy.

\textbf{High-level strategy}. The ultimate objective of qubitizing the double-factorized Hamiltonian, as with naïve qubitization, is to selectively apply indexed terms of the Hamiltonian, weighted by the appropriate coefficients in the Hamiltonian sum. However (as noted in Eq.~\eqref{eq:ham_rewrite}), the terms we would like to selectively apply in the double-factorized Hamiltonian are non-trivial basis-rotated Majorana operators. When we perform the Hamiltonian factorization, we are performing an eigendecomposition which leaves us in the eigenbasis of each $L^{(r)}$ matrix, hence ``basis-rotated" operators. The amendments made to the naïve qubitization circuit are due to the machinery needed to load basis-rotating angles onto an ancilla register, and then selectively apply a series of rotations onto the system register conditioned on this ancillary register. These rotations are necessary in order to rotate the system qubits \textit{out} of the eigenbasis of each $L$ matrix, to then apply a standard Majorana operator (a Pauli), and then rotate \textit{back} into the original basis. Loading and then applying these angles requires two ancilla registers, one of size $\lceil \log R \rceil$ and another of size $\lceil \log M \rceil$, to serve as index registers over the ranges $r \in R$ and $m \in M$ respectively (and a host of additional workhorse ancillary registers).
Following Figure 16 in Ref.~\cite{Lee2020} (but using the notation of Ref.~\cite{DoubleFactorized_MSFT}), one can interpret the initial series of state preparation routines and data-loaders as one large $\textsc{prepare}$ subroutine that loads the appropriate rotation angles, and then the sequence of rotations (and its dagger) and the applied Pauli as one large $\textsc{select}$.

\textbf{Dominant subroutines}. We note that the plurality of the T-count comes from applying the series of rotations and a data-loader that loads the angles for these rotations, and a vast majority of the total qubits are used in these two steps (summarized for each molecule and basis set in Table~\ref{table:subroutine_split}). Here we highlight that the brunt of the asymptotic expressions in the overall algorithm's resource estimate comes from the gate and qubit count of the basis-rotation $\text{QROM}$. Parallelizing the T-count of $\text{QROM}$ in Eq.~\eqref{eq:qrom_cost}, the list of elements we must load is a list of $M + N$ angles, each approximated by $N \cdot \beta$ bits of precision, leading to an asymptotic T-count:
\begin{equation} \label{eq:rotation_qrom_cost}
\left\lceil{\frac{M + N}{\lambda}}\right\rceil + N\beta \cdot (\lambda - 1)
\end{equation}
While each $\text{QROM}$ in the circuit has its own tunable parameter $\lambda$, the $\lambda$ for the basis-rotation $\text{QROM}$ contributes the largest number of ancillas of all of the data-loaders, and can be taken to be the same proxy $\lambda$ shown in the algorithm's asymptotic scaling expression in Eq. \eqref{eq:asymptotic_with_lambda}.

\section{Details on algorithmic parallelization techniques}
\label{sec:algo_parallelization}
While we typically take the T-count $n_T$ to be a measure of the runtime of the algorithm, following the numerics in Fig.~\ref{fig:MSD_footprint_data} showing the relatively small size of the magic state factory, one could instead take the T-depth $D_T$ as this measure. The T-depth of an algorithm is a proxy for how parallelizable the algorithm is; given the ability to perform multiple T-gates in parallel (see Section~\ref{sec:constant_time}) and having $\sim(n_T/D_T)$ magic state factories available, one can perform up to $n_T$ gates in $D_T$ time slices.

In these circumstances where multiple magic states may be injected and consumed in parallel, it may be preferable to take $V_D$ as the measure of computational volume rather than $V_n$, which assumes the \emph{serial} application of T gates. For the qubitization algorithm described in this paper, there are a handful of techniques and augmentations to subroutines one can exploit to realize a $V_D$ that is significantly smaller than $V_n$. These techniques range from known, trivial techniques, to novel and/or non-trivial optimizations. First we comment on the two largest-impact non-trivial augmentations, and in the following subsection we comment on the the other techniques considered. Together, the two subroutines discussed below constitute the majority of the overall computational volume of the algorithm (see Table~\ref{table:subroutine_split} for percent contributions per molecule and basis set).

\subsection{Parallelization of the two dominant subroutines}
\label{sec:parallelization}

\textbf{$\textsc{QROM}$ depth vs. $\textsc{QROM}$ count} The data-lookup subroutine is used throughout the qubitized double-factorization circuit; however, a particular $\textsc{QROM}$ used for loading angles for basis-changing rotations (the next subroutine considered) constitutes a large portion of the overall T-cost of the algorithm. For this particular $\textsc{QROM}$, we must load a list of $M + N$ angles, each approximated by $N \cdot \beta$ bits of precision, leading to the asymptotic T-count given in Eq. \eqref{eq:rotation_qrom_cost}. While each $\text{QROM}$ in the circuit has its own tunable parameter $\lambda$, the $\lambda$ for the basis-rotation $\text{QROM}$ contributes the largest number of ancillas of all of the data-loaders, and can be taken to be the same proxy $\lambda$ shown in Eq. \eqref{eq:asymptotic_with_lambda}.

The T-count of data-loaders has a linear dependence on both the number of items being loaded, and the bits of precision used to approximate each element (see Eqs.~\eqref{eq:qrom_cost} and \eqref{eq:rotation_qrom_cost}). On the other hand, the depth of a data-loader has linear dependence on the number of items being loaded, but \textbf{no} dependence on the bits of precision~\cite{LowKliuchnikov}. For the data-loader that loads the basis-changing rotations, we can compare the asymptotic T-count given in Eq. \eqref{eq:rotation_qrom_cost} with the asymptotic T-depth given below:
\begin{equation} \label{eq:qrom_depth}
  \ceil[\Big]{\frac{M+N}{\lambda}} + \log{\lambda}.
\end{equation}
This drastic difference results from one of $\textsc{QROM}$'s primitives, the so-called $\textsc{SwapUp}$ network~\cite{LowKliuchnikov,Wan_2021}. We can use this improvement throughout the circuit in all of the standalone $\textsc{QROM}$s and the ones inside of each $\textsc{prepare}$. However, this lack of dependence on the bits of precision is particularly significant in the context of the large data-loaders mentioned above––which make up a plurality of the computational volume (see Table~\ref{table:subroutine_split})––because the bits of precision for these data-loaders is $N \beta$. Removing this dependence can drastically reduce the overall T-depth.

\textbf{Log-depth basis rotations} Like the previous savings contribution, this savings results from exploiting a drastic difference between the T-count and T-depth of one of the largest percentage share subroutines in the qubitization circuit. Here we introduce a novel construction of the basis-changing rotations (detailed in Appendix VII.C.1 and VII.C.2 in \cite{DoubleFactorized_MSFT}) that achieves logarithmic depth in the number of target qubits, as opposed to linear depth, \textit{without} the need for extra ancillae, and \textit{without} altering the T-count. Since this subroutine contributes the next largest amount of computational volume to the overall circuit after the $\textsc{QROM}$ mentioned above  (see Table~\ref{table:subroutine_split}), and reduces the T-cost of this circuit from $\mathcal{O}(N)$ to $\mathcal{O}(\log{N})$, this aids in drastically reducing the overall computational volume.

The circuit for implementing block-encodings of products of basis-rotated Majorana fermions is described in Lemma 8 and Eqs. (57, 58) of~\cite{DoubleFactorized_MSFT}. The dominant cost of the circuit is a basis rotation that converts a ``rotated'' Majorana fermion $\gamma_{\vec{u}} = \sum_{j \in [N]} u_j \gamma_j$ to the ``unrotated'' Majorana $\gamma_0$, which can be readily implemented by a single Pauli under the Jordan-Wigner encoding. In \cite{DoubleFactorized_MSFT}, Lemma 8, it is shown that the fermionic basis rotation can be constructed from a ladder of operators of the form $V_p = \exp\{\theta_p \gamma_p \gamma_{p+1}\}$, where each angle in the set $\{\theta_p\}$ is chosen to zero a single term in the basis-rotated Majorana $\gamma_{u}$. Under the Jordan-Wigner encoding, $V_p$ becomes a Givens rotation. A circuit diagram is reproduced in Fig. \ref{fig:gizens} from \cite{DoubleFactorized_MSFT}, Eq. (58).

\begin{figure}[h]
	\includegraphics[width=0.30\textwidth]{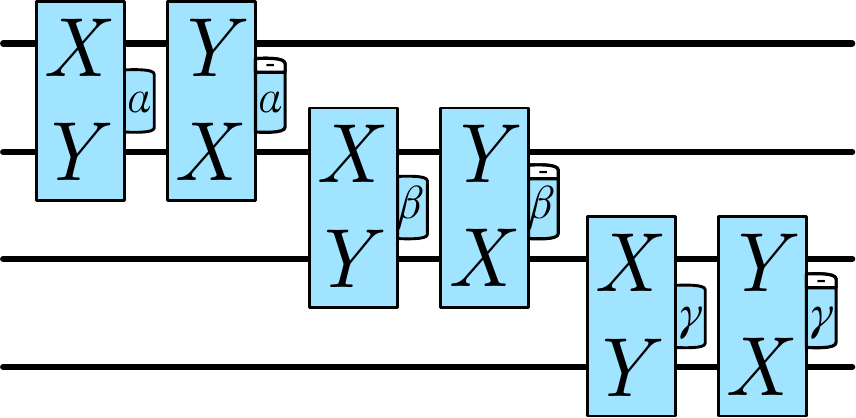} \quad
	\includegraphics[width=0.29\textwidth]{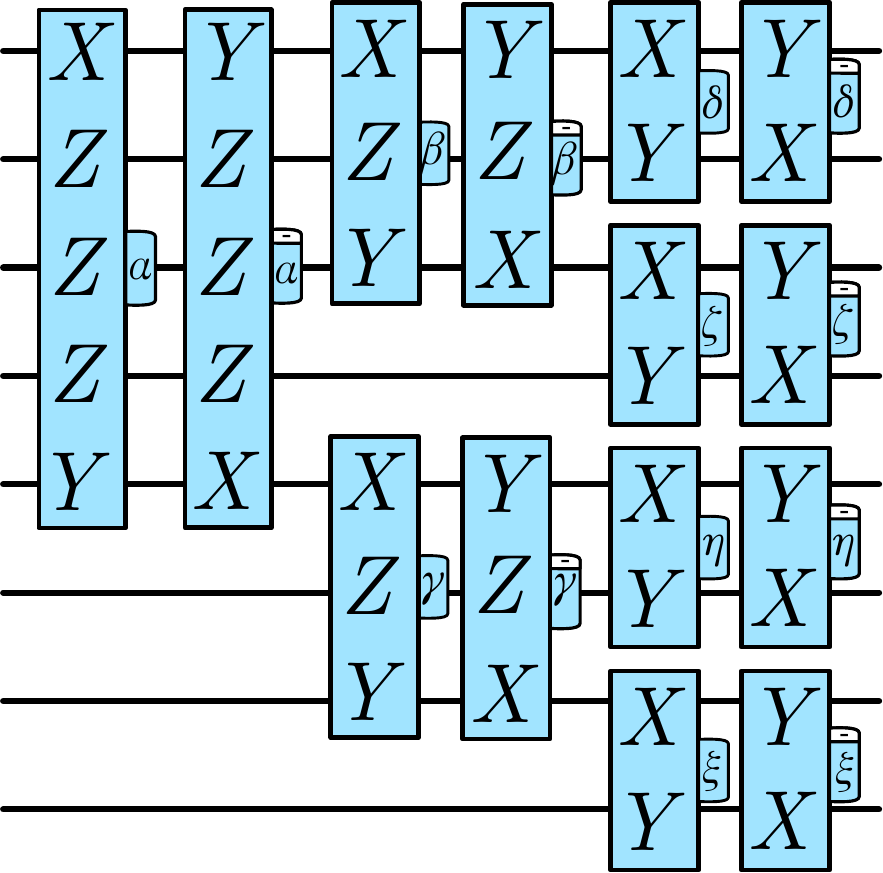}
	\caption{(left) A circuit diagram of the basis rotation circuit from \cite{DoubleFactorized_MSFT} for $N=8$. Each Givens rotation is a pair of commuting PPRs with opposite angle, and therefore has rotation count 2 and rotation depth 1. (right) A circuit diagram of the Gizens basis rotation circuit for $N=8$. Each Gizens rotation is also a pair of commuting PPRs with opposite angle, and therefore has rotation count 2 and rotation depth 1.}
	\label{fig:gizens}
\end{figure}

It is somewhat trivial to note that the basis change circuit requires $N - 1$ Givens rotations to zero $N - 1$ parameters in $\gamma_{\vec{u}}$, and so we cannot reduce the rotation count any further. However, the overlapping Givens rotations in Fig \ref{fig:gizens} do not commute and so the rotation depth in this instance scales identically to the rotation count. This need not be the case. Rather than $V_p$, we can construct a circuit from the generalized fermionic operator $\tilde{V}_{p,q} = \exp\{\theta\gamma_p \gamma_q\}$. Under the Jordan-Wigner encoding, these operators become:
\begin{align}
  \tilde{V}_{p,q}(\theta) \rightarrow &\exp\left\{\theta\left(X_p \otimes \left(\bigotimes_{p<t<q} Z_t\right) \otimes Y_q\right)\right\} \nonumber \\
  \cdot&\exp\left\{\theta\left(Y_p \otimes \left(\bigotimes_{p<t<q} Z_t\right) \otimes X_q\right)\right\}.
\end{align}
Given the inserted $Z$s, we'll refer to this operator as a ``Gizens'' rotation.

The approach is similar to that used for state preparation by Grover and Rudolph in \cite{groverrudolph}. Given the ability to compute partial sums of the amplitudes $u_p$, we apply a binary tree of Gizens rotations with angles chosen such that the correct partial sum is obtained at either side of the bifurcation. An example circuit is shown in Fig \ref{fig:gizens}. The proof that this construction evaluates the correct basis change circuit is the same as in \cite{groverrudolph}; we include it in Appendix~\ref{sec:gizens_proof}.

\subsection{Further parallelization techniques}

\textbf{T-depth vs. T-count of a Toffoli} A trivial speedup that can be implemented is in how we account for the T cost of Toffoli gates. In the resource estimates provided in Refs.~\cite{DoubleFactorized_MSFT,Lee2020}, the cost of the subroutines is given in terms of the number of Toffoli gates, whereas we provide the cost in terms of T-gates. One can implement a Toffoli using $4$ T-gates ~\cite{Jones2013}, and so all of our gate count expressions are a factor of $4$ larger than those provided in Refs.~\cite{DoubleFactorized_MSFT,Lee2020}. However, this factor of $4$ can be removed by using the T-depth of a Toffoli instead of its count, which is known to be only depth-$1$ \cite{Jones2013,Selinger_2013}. The $4$T-gate, depth-$1$ Jones Toffoli \cite{Jones2013} only uses a single clean ancilla (which can be borrowed from other clean ancillae in the algorithm, and re-used to decompose each Toffoli in the computation). Each Gidney compute-uncompute pair used in the $\textsc{select}$s within $\text{QROM}$s first introduced in Ref.~\cite{Babbush2018} can similarly be parallelized to depth-$1$ using a modification of the Jones construction in Ref.~\cite{Jones2013}.

\textbf{Controlled swaps in depth-$1$} For the double-factorized qubitization circuit depicted in Figure 16 in Appendix C of Ref.~\cite{Lee2020}, we must perform a series of controlled swaps a total of $4$ times, once before and once after each round of the basis-changing rotations detailed in Appendix \ref{sec:parallelization}. Each sequence of controlled swaps requires the application of $N$ Toffoli gates, where $N$ is the number of orbitals; however, a sequence of swaps targeting disjoint sets of qubits but with control on the \textit{same} control qubit can be implemented in $\mathcal{O}(1)$ depth~\cite{LowKliuchnikov,Campbell_2017}, taking the T cost from a Toffoli count of $4N$ to a constant T-depth independent of $N$. This low depth of sequences of controlled swaps can also be exploited within each $\textsc{prepare}$ subroutine used in the circuit, as the ``alias sampling" version of $\textsc{prepare}$ ends with a sequence of controlled swaps~\cite{Babbush2018}. Achieving this depth requires no additional ancillae.

\textbf{Log-depth comparators} Each ``alias sampling" version of $\textsc{prepare}$ also uses a comparator to perform an inequality test between two registers of qubits. Comparators are made up of adders; the trade-off between count and depth of adders is detailed in Ref.~\cite{gidney2020quantum}. The T-count of most adders have linear dependence in the size of the addend registers, and depth-efficient adders tend to have \textit{worse} constants than count-efficient adders. However, in the regime where performing T-gates in parallel is cheap, the time-savings from log-depth adders over linear-depth adders dwarfs the constant disadvantage in count. Our battery chemistry application falls squarely in this regime, so we opt for low-depth adders in our comparator constructions. The required number of ancillae used for either the linear or logarithmic depth comparator is linear in the register size~\cite{gidney2020quantum,LogDepthAdder}.

\subsection{Optimizing for $V_D$ vs. $V_n$}

When considering T-depth $D_T$ in addition to T-count $n_T$ and qubit count $n_L$, one must decide which computational volume, $V_n$ or $V_D$ to optimize. In order to get a sense of the magnitude of computational volume improvements possible, it is worth considering the resource costs of three scenarios: optimizing for $V_n$ and applying T gates serially; optimizing for $V_n$ but applying multiple T gates in parallel; and optimizing for $V_D$ and applying T gates in parallel. We provide tables with the resource estimates for each molecule and basis set (Tables~\ref{table:double_tcount}, \ref{table:tdepth_savings}, and \ref{table:tdepth_savings_depth_opt}), but here for each scenario we quote the smallest and largest classically intractable instances considered: full configuration interaction (FCI) picture of $\text{PF}_6^{-}$ in the cc-pVDZ basis as the smallest instance, and FCI picture of LFEC in the cc-pVTZ basis as the largest.

\textbf{Serial execution of T gates} As an initial starting point (and to make a fair comparison with the estimates of prior studies \cite{DoubleFactorized_MSFT,Lee2020}), it is informative to first note the algorithmic resource estimates in the case where we apply T gates serially; that is, optimizing for $V_n$ and also taking $V_n$ to be the computational volume. The T-counts and qubit counts for all molecules and basis sets are provided in Table~\ref{table:double_tcount}. For the smallest (largest) instance we have: $n_L = 16,\!382 \: (100,\!139)$ qubits and $n_T = 2.32 \times 10^{12} \: (1.06 \times 10^{14})$ T gates. Thus, the computational volume for serial application of T gates ranges from $\mathcal{O}(10^{16})$ to $\mathcal{O}(10^{19})$.

\textbf{Swapping count for depth} It is also informative to see what kind of volume savings are possible by merely swapping out $n_T$ for $D_T$ in the volume estimate; that is, minimizing $V_n$ but taking the resulting $V_D$ to be the computational volume. Performing this naïve substitution (without consideration of depth-optimization) exploiting the difference between depth and count for the subroutines highlighted above, we can achieve an overall computational volume reduction of $\sim 8$x for the smallest instance, $\sim 13$x for the largest instance, and $\sim 10$x for most instances. The detailed savings for each molecule and basis set are reproduced in Table~\ref{table:tdepth_savings}. Thus, the computational volume for parallel application of T gates while minimizing T-count $n_T$ ranges from $\mathcal{O}(10^{15})$ to $\mathcal{O}(10^{18})$.

\textbf{Depth optimization} Given that we are in a regime where we can opt to use $V_D$ as our computational volume, it makes sense to see how far we can push the potential savings by optimizing for $V_D$ instead of $V_n$. Optimizing for $V_D$ (alternatively, performing the minimization prescribed in Eq. \eqref{eq:cost_function_depth}, we ultimately find that the potential volume reduction is \textit{slightly greater} than that offered by simply swapping out count for depth, resulting in savings between $\sim 9$x and $\sim 15$x . This is shown in detail in Table~\ref{table:tdepth_savings_depth_opt}. Note that no qubit counts are provided as (remarkably) the $V_D$-optimal volume uses the same number of qubits as the $V_n$-optimal volume.

The optimization of $V_D$ consists of choosing an optimal set of parameters $\lambda$ for each $\textsc{QROM}$ with respect to this metric. As shown in Table~\ref{table:subroutine_split}, a plurality of the computational cost of the algorithm comes from performing a large $\text{QROM}$; however, it is also noted that $\textsc{QROM}$s are used throughout the qubitization procedure, both as standalone subroutines and also as part of each $\textsc{prepare}$ subroutine encountered in the circuit. Each of these data-loaders has its own $\lambda$ that can be optimized either \emph{independently} of or \emph{dependently} on each other.

\subsection{Various $\lambda$ optimization strategies}
\label{sec:lambda}
The choice of $\lambda$ greatly influences the overall resource estimates. Tuning the parameter $\lambda$ for any specific $\textsc{QROM}$ results in a trade-off between T-count and qubit count~\cite{LowKliuchnikov}; if it is large, the T-count is reduced at the expense of increasing the number of qubits, and vice-versa. However, another trade-off noted in Ref. ~\cite{LowKliuchnikov} is trading depth for count: one can continue to increase $\lambda$ past the point of minimal T-count, and decrease the T-depth at the expense of the number of qubits \textit{and} the T-count. The choice of $\lambda$ for each data-loader in the algorithm depends on whether one is trying to minimize the computational volume given as $V_n = n_T \times n_L$ or given as $V_D = n_D \times n_L$.

\textbf{Minimizing $V_n$} A natural choice in order to make a fair comparison with the methods of Refs.~\cite{Babbush2018,Berry2019,DoubleFactorized_MSFT,Lee2020} (thus, assuming the serial application of T gates) is to choose $\lambda$ to minimize $n_T \times n_L$. The minimum is achieved by  $\lambda = \mathcal{O}(\sqrt{M/N\beta})$, which leads to the following asymptotic expressions:
\begin{equation} \label{eq:asymptotic}
  \begin{aligned}
    n_{T,Q}&= \mathcal{O}(\sqrt{MN\beta}+N\beta), \\
    n_L&= \mathcal{O}(\sqrt{MN\beta}+N+\beta).
  \end{aligned}
\end{equation}
The exact numbers for the chemical systems introduced in Sections \ref{sec:chem_background} and \ref{sec:comp_description} are shown in Table~\ref{table:double_tcount}.

For clarity, we note two points about these resource estimates. First, the resource estimates for the molecules and basis sets that we consider are generally \emph{larger} than those provided in Refs.~\cite{DoubleFactorized_MSFT,Lee2020} (especially for the larger basis sets cc-pVDZ and cc-pVTZ). This is to be expected, as the factorization for the Hamiltonians of these molecules yield larger values for the ranks $R$ and $M$, and the norm $\alpha$ than those for molecules the size of FeMoco; additionally, those references consider the Toffoli count whereas here we consider the T-count (which is a factor of $\sim 4$x larger). Second, we note that the values in Table~\ref{table:double_tcount} assume the \emph{serial} application of magic states (and so, these estimates do \emph{not} make use of the algorithmic and architectural improvements we outline in Sections \ref{sec:algorithmic_parallelization} and \ref{sec:results_ppms}). In order to achieve the optimal space-time costs for these computations, it is \textit{not} obvious that one should opt to minimize $n_T \times n_L$. As shown in Fig.~\ref{fig:MSD_footprint_data}, for the T-counts and qubit counts quoted in Table~\ref{table:double_tcount}, the estimated size of the magic state factories necessary to distill T-gates is remarkably small. In such a regime, it may be preferable to minimize $D_T \times n_L$ and parallelize the application of magic states.

\textbf{Minimizing $V_D$} The optimization of $V_D$ consists of choosing an optimal set of parameters $\lambda$ with respect to this metric. Recall that a plurality of the computational cost of the algorithm comes from performing a large $\text{QROM}$; however, it is also noted that $\textsc{QROM}$s are used throughout the qubitization procedure, both as standalone subroutines and also as part of each $\textsc{prepare}$ subroutine encountered in the circuit. Each of these data-loaders has its own $\lambda$ that can be optimized \textit{independently} of each other. Previously, this independent optimization was done in order to minimize $V_n$ for each of these subroutines; a natural consideration is to instead choose each $\lambda$ to minimize $V_D$ for each $\textsc{QROM}$.

However, this naïve consideration does \textit{not} lead to the most optimal computational volume. Instead, to achieve the smallest volumes shown in Table~\ref{table:tdepth_savings_depth_opt}, we found that it was most beneficial to optimize each $\lambda$ \textit{contingent} on all other $\lambda$ values. The strategy is to determine which $\text{QROM}$, with $\lambda$ chosen to minimize Eq. \eqref{eq:qrom_cost}, contributes the largest number of ancillary qubits:
\begin{equation}
 Q_{\text{max}} \coloneqq \lambda_{\text{min-count}} \times \text{bits-of-precision}
 \end{equation}
 We then assign all other $\lambda$s to each $\text{QROM}$ by dividing $Q_{\text{max}}$ by the bits of precision for each $\text{QROM}$.

By finding the maximum number of ancillae used in a data-loader $Q_{\text{max}}$, this ensures that we have enough qubits allocated to borrow from in order to later set $\lambda$s for the rest of the $\text{QROM}$s that may indeed be \textit{larger} than their $\lambda_{\text{min-count}}$ value. It is noted in Ref.~\cite{LowKliuchnikov} that one can increase the value of $\lambda$ in a data-loader greater than its $\lambda_{\text{min-count}}$ that minimizes Eq. \eqref{eq:qrom_cost} in order to minimize depth at the expense of count. This is exactly what we do in order to minimize $V_D$; note that the counts reproduced in Table~\ref{table:tdepth_savings_depth_opt} exceed those reproduced in Table~\ref{table:tdepth_savings} (hence the enormous ratios in the penultimate column of Table~\ref{table:tdepth_savings_depth_opt}).

\textbf{Flexibility} The conclusion of the two preceding subsections is \textit{not} that one should necessarily optimize $V_D$ instead of $V_n$. This decision depends not only on the absolute smallest computational volume, but also on a number of other factors. Depending on one's architecture, one may employ alternative optimization strategies.

To give two brief examples, one could instead determine $\lambda_{\text{min-depth}}$ values that minimize the $V_D$ of each $\textsc{QROM}$. Overall, this leads to similar decreases in computational volume, but a larger T-depth and approximately an order of magnitude \textit{reduction} in qubit count compared to the values reproduced in Table~\ref{table:tdepth_savings_depth_opt}. As another example, opting for the strategy for minimizing $V_n$ may be preferable to that for optimizing $V_D$ if one is concerned about the overall number of magic state factories on your device. One would need on average $(\text{T-count}/\text{T-depth})$ many magic state factories in order to fully exploit the speedups presented in these sections; this ratio is markedly large, between $20$ and $50$, for the strategy that minimizes $V_D$ (shown in detail for each molecule in the penultimate column of Table~\ref{table:tdepth_savings_depth_opt}).

When assessing potential volume reductions by considering T-depth in addition to T-count, there is a good deal of flexibility depending on the relative importance one gives to qubit count, number of magic state factories, etc.. All of these potential savings are close in magnitude, so it is not obvious which optimization strategy one should opt for. The underlying architecture of a given approach may best inform which strategy to use. As yet, it is unclear how the transversal gates used for constant-time PPMs will affect the threshold and sub-threshold scaling of the codes used, and the potential runtime speedups will greatly depend on the timescales one expects of physical operations on a specific platform; we leave these studies for future work.

\section{Constant-time PPM details}
\label{sec:constant_time}
\subsection{Fast injection of distilled magic states}
\label{sec:faster_magic}

The resource estimate that produced the numbers in Table \ref{table:double_tcount} was based on the sub-threshold scaling behavior of the FBQC scheme in Ref.~\cite{FBQCpaper}, making use of one of the quantum computing architectures in Ref.~\cite{litinski2019game}. We propose an alternative architecture in this section, in the hope that it may reduce the overall space-time cost of the algorithm. As we shall later see, we have reasons to believe that the newly proposed architecture can reduce the computation time without incurring a significant amount of extra footprint. However, let us emphasize that, unlike the results above, we have \emph{not} carried out a Monte Carlo study on the threshold and sub-threshold behavior of this new scheme. A fair comparison between the two can be made only via a rigorous numerical study, which we leave for future work.

We previously observed that if we use a state-of-the-art quantum algorithm to simulate the molecules relevant to battery research, the magic state factory constitutes only up to $2\%$ of the entire quantum computer. Therefore, if magic state distillation is the main bottleneck of the quantum computation, one may be able to reduce the computation time by increasing the number of magic state factories, and injecting the distilled magic states appropriately.

Optimal injection of multiple distilled magic states is a nontrivial problem, since what is best can depend on a number of different factors, including the specific magic state distillation protocol used, routing of qubits, and the way in which Clifford gates are implemented. One possibility is to simply distill the $\ket{T} = (\ket{0} + e^{i\pi/4}\ket{1})/\sqrt{2}$ state and inject it. Alternatively, one can distill a $\ket{CCZ} = CCZ\ket{+++}$ state~\cite{Jones2013,Gidney2019efficientmagicstate,Chamberland2020}, which can be readily and efficiently used in the data-loader described in the discussion on parallelization above~\cite{DoubleFactorized_MSFT,Lee2020}.

In this section, we propose a different scheme which is based primarily on Litinski's scheme~\cite{litinski2019game,litinski2019magic}. Litinski's scheme uses Pauli Product Measurement (PPM) and preparation of states in the set $\{ |0\rangle, |+\rangle, |T\rangle\}$. PPM refers to a non-destructive measurement of the following observable:
\begin{equation}
  \bigotimes_{n=1}^{N} P_n, \label{eq:PPM}
\end{equation}
where $P_{n} \in \{I, X, Y, Z \}$ and $P_n$ is the logical Pauli operator of the $n$'th qubit, which ranges from $1$ to $N$. Preparation of $|0\rangle$ and $|+\rangle$ can be implemented in $\mathcal{O}(1)$ time. On the other hand, the time to prepare $|T\rangle$ depends on the choice and the number of magic state factories. Performing a PPM takes $\mathcal{O}(d)$ time, where $d$ is the code distance. Therefore, if the number of magic state factories is abundant, the main bottleneck becomes the speed of the PPM. Using Fowler’s technique \cite{fowler2013timeoptimal}, one can speed up the computation. However, in this scheme, a $k$-fold increase in speed necessitates a $k$-fold increase in footprint.

The main result of this section is an $\mathcal{O}(1)$-time implementation of an arbitrary PPM which does not incur such additional footprint. Moreover, all the physical operations can be made local in two spatial dimensions. Therefore, optimistically speaking, in the regime in which the number of $|T\rangle$-state factories is abundant so that more than one $|T\rangle$ state can be distilled in a single logical clock cycle, one can expect to be able to inject them all in a single logical clock cycle, provided that there are no more than $\mathcal{O}(d)$ of them. The rest of this section is structured as follows. We: (a) discuss an abstract circuit model that implements a PPM; (b) explain how this protocol can be implemented on the surface code~\cite{bravyi1998quantum}; (c) explain how this method can be used to compile a general quantum algorithm; and (d) discuss the expected speedup using this approach. We remark that while the fault-tolerant protocols in this section are described in terms of surface codes (for simplicity and familiarity), they may readily be expressed in FBQC in terms of the 6-ring network of Section \ref{sec:methods_overhead}.

\subsection{PPM Circuit}
\label{sec:ppm_circuit}
To understand our protocol, it is helpful to first consider an abstract circuit model which captures its spirit. Suppose we have $N$ qubits and we would like to implement a PPM. Without loss of generality, we can write down a Pauli as follows:
\begin{equation}
  \begin{aligned}
    P &= \bigotimes_{j=1}^N P_j \\
    &= i^s \left(\bigotimes_{j=1}^N X_j^{m_j^{X}} \right) \left(\bigotimes_{j=1}^{N} Z_j^{m_j^{Z}} \right),
  \end{aligned}
\end{equation}
where $s\in \{0,1,2,3\}$ and $m_j^X, m_j^Z \in \{0, 1\}$.

Hypothetically, if we were to measure individual $P_j$s, we could use an ancillary qubit initialized in the $|+\rangle$ state and apply the following entangling gates:
\begin{equation}
  \begin{aligned}
    P_j&=X: \quad CX \\
    P_j&=Y: \quad (S^{\dagger}\otimes I)\cdot CZ\cdot CX \\
    P_j&=Z: \quad CZ,
  \end{aligned}
  \label{eq:ppm_recipe}
\end{equation}
where the control is the ancillary $|+\rangle$ state and the target is the qubit being measured. Here the $S^{\dagger}$ gate acts on the control qubit. The measurement is performed in the $X$ basis. One can verify that $Z$ is applied to the ancillary state if and only if the qubit being measured is in the $+1$ eigenstate of the Pauli chosen.

Of course, what we really want is to measure $P$ without measuring the individual $P_j$s explicitly. For that purpose, one can use the $|\text{CAT}\rangle$ state, defined as follows:
\begin{equation}
  |\text{CAT}\rangle_{N} = \frac{1}{\sqrt{2}}\left(|\underbrace{0\ldots 0}_N\rangle + |\underbrace{1\ldots 1}_{N}\rangle \right).
\end{equation}
Now, one can pair up the $j$'th qubit that forms the $|\text{CAT}\rangle$ state to the $j$'th data qubit and apply the gates in Eq.~\eqref{eq:ppm_recipe}. If we measure the tensor product of $X$ on the Paulis, the parity of those measurement outcomes is precisely the non-destructive measurement outcome of $P$; see Fig. \ref{fig:ppm_circuit} for an example.

\begin{figure}[h]
  \begin{quantikz}
    \lstick[wires=4]{$|\text{CAT}\rangle_4$} &\ctrl{4} & \qw & \qw &\qw &\qw &\meter{X}\\
    & \qw & \ctrl{4} & \ctrl{4} &\gate{S^{\dagger}} &\qw &\meter{X}\\
    & \qw & \qw & \qw &\qw & \ctrl{4} &\meter{X}\\
    & \qw & \qw & \qw &\qw &\qw &\meter{X}\\
    \lstick[wires=4]{Data} & \targ{} &\qw  & \qw &\qw &\qw & \qw \\
    & \qw & \targ{} & \ctrl{} &\qw &\qw & \qw \\
    & \qw & \qw & \qw &\qw & \ctrl{}& \qw \\
    & \qw & \qw & \qw &\qw &\qw & \qw
  \end{quantikz}
  \caption{A circuit-level implementation of a PPM for a Pauli string $P=X\otimes Y\otimes Z\otimes I$. The outcome of the PPM is the joint parity of each depicted measurement's outcome.\label{fig:ppm_circuit}}
\end{figure}

\subsection{Fast PPM for surface code} 
\label{sec:fast_ppm}
Now, let us take a step back from this abstract circuit picture and think about how one can implement all the procedures we have explained so far fault-tolerantly, in particular, using patches of surface-codes~\cite{bravyi1998quantum}. Recall that the surface code is defined on a lattice with open boundary conditions, with two ``smooth'' and two ``rough'' boundaries. The qubits are arranged on the edges of the lattice. The surface code is a stabilizer code with a stabilizer group generated by a set of ``star'' operators, which are tensor products of $X$s along the edges that are incident on a vertex, and a set of ``plaquette'' operators, which are tensor products of $Z$s along the edges that surround the plaquette; see Fig. \ref{fig:surface_code}.

\begin{figure}[h]
  \includegraphics[width=0.5\columnwidth]{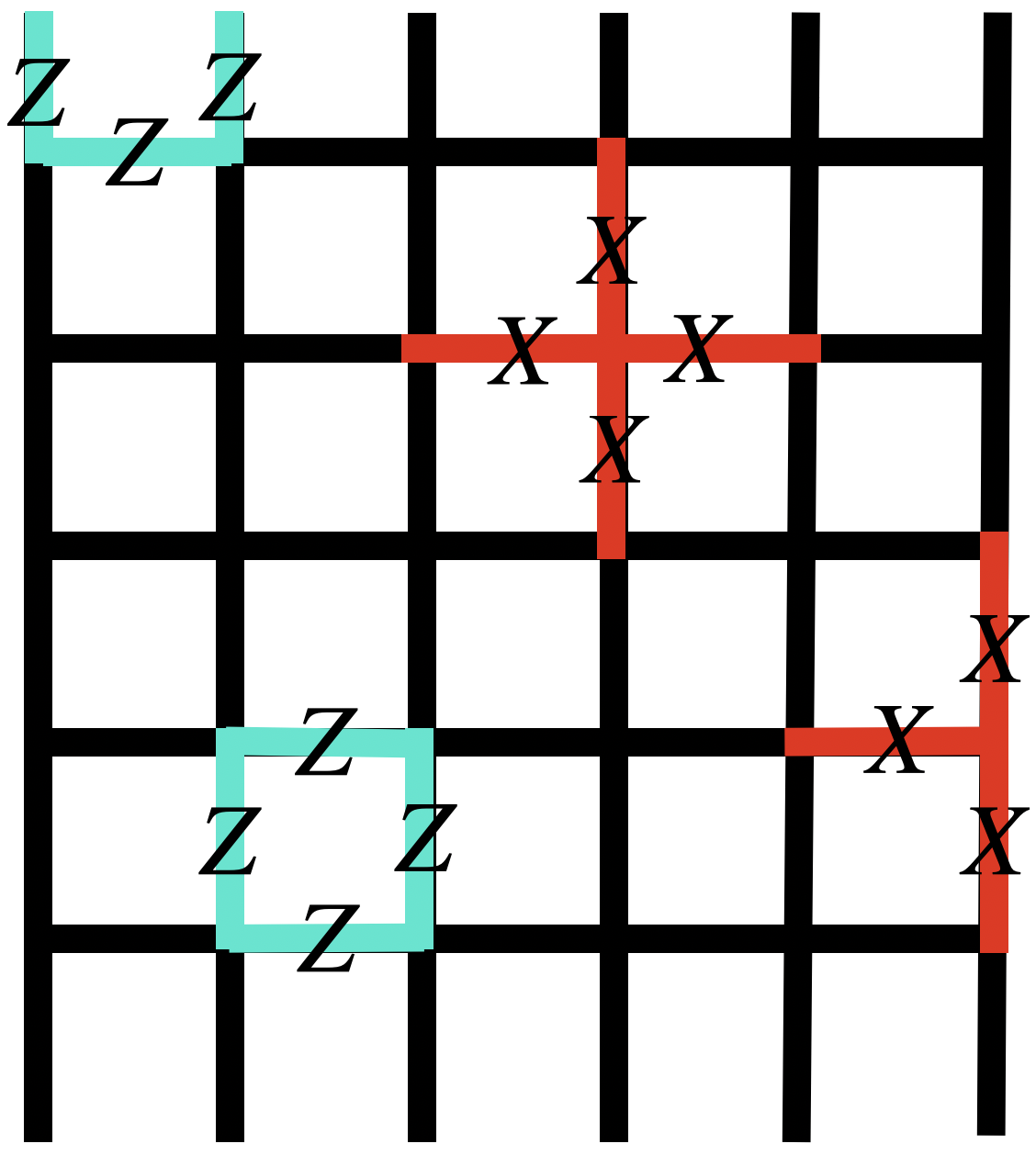}
  \caption{Surface code and its stabilizers. In this figure, the code distance is $d=6$.\label{fig:surface_code}}
\end{figure}

Since our goal is to implement the circuit in Fig. \ref{fig:ppm_circuit} in constant time, we must understand whether the gates that appear in this circuit diagram can be implemented in constant time on the surface code. We shall see towards the end of this section that preparation of an analogue of the $|\text{CAT}\rangle_N$ state (as opposed to the exact $|\text{CAT}\rangle_N$ state over the logical qubits) and measurement in the $X$-basis can be done in constant time. So let us focus instead on whether we can implement the unitary gates in Eq.~\eqref{eq:ppm_recipe}.

Because the surface code is a CSS code, it allows for a straightforward implementation of a transversal CNOT gate. Moreover, the CZ gate can be viewed as a CNOT gate conjugated by Hadamards. By folding a surface code in half along its diagonal line, it is possible to implement the logical Hadamard by applying a transversal Hadamard gate on all the qubits followed by a local relabeling of the qubits~\cite{Moussa2016} (which builds upon the local unitary equivalence between the the color code and surface code~\cite{KYP}). We remark that this folding operation is accessible in an interleaved architecture~\cite{Interleaving}.  Therefore, the CZ gate can also be implemented in constant time. Implementation of $S$ can also be done in constant time, either by injecting a $+1$ eigenstate of the Pauli-$Y$ operator, or by applying a transversal gate in the folded surface code picture. The transversal gate option, however, requires an extra entangling gate which may adversely affect the threshold and the sub-threshold scaling. On the other hand, the injection requires an additional logical qubit for \emph{every} logical data qubit, incurring up to a $50\%$ increase in the footprint.

Fortunately, it is possible to have the best of both worlds. Below, we explain a method to modify Eq.~\eqref{eq:ppm_recipe} so that one does not need to apply the $S^{\dagger}$ gate explicitly. The price we have to pay is modest; we only need two additional logical qubits, one that keeps the $+1$ eigenstate of $Y$ and another one dedicated to increasing the size of the $|\text{CAT}\rangle$ state to $N+1$.

The key idea is to \emph{not} apply the $S^{\dagger}$ explicitly, but instead update the observable that needs to be measured. In this case, the Pauli string simply becomes a tensor product of $X$s and $Y$s. In the surface code, measurement in the $X$ basis is easy, but measurement in the $Y$-basis is less straightforward; the former can be done in constant time, but to the best of our knowledge, measurement in the $Y$ basis requires an implementation of a single-qubit Clifford gate such as $H$ or $S$. However, this Pauli string (say $Q$), up to the stabilizer of the $|\text{CAT}\rangle$ state, is ``almost'' equal to a Pauli string consisting purely of $X$s. An important observation is that the $|\text{CAT}\rangle_{N+1}$ state obeys the following stabilizer constraint, both before and after the gates in Eq.~\eqref{eq:ppm_recipe} are applied:
\begin{equation}
  Z_nZ_m |\text{CAT}\rangle_{N+1} = |\text{CAT}\rangle_{N+1}
\end{equation}
for all $n, m \in \{1,\ldots, N+1 \}$. Therefore, if $Q$ has an even number of $Y$s,
\begin{equation}
  Q |\text{CAT}\rangle_{N+1} = \left(\bigotimes_{j=1}^N X_j\right)|\text{CAT}\rangle_{N+1} (-1)^{n_Y/2},
\end{equation}
where $n_Y$ is the number of $Y$s in $Q$. If $n_Y=0\mod 2$, up to a phase, the joint parity of $Q$ can be measured by measuring all the qubits in the $X$-basis, an operation that can be performed in constant time. If $n_Y$ is odd, then there will be one leftover $Y$. In that case, one can apply the gate corresponding to the $P_j=Y$ case for the $j=N+1$'th data qubit, which will be initialized to the $+1$ eigenstate of $Y$. This makes the last Pauli in $Q$ equal to $Y$, making $n_Y$ even. Moreover, the $+1$ eigenstate of $Y$ is returned to its original state after this process, which can be reused later; see Fig.~\ref{fig:ppm_circuit_nophase} for an example.

\begin{figure}[h]
  \begin{quantikz}
    \lstick[wires=5]{$|\text{CAT}\rangle_5$} &\ctrl{5} & \qw & \qw &\qw &\qw &\qw &\meter{X}\\
    & \qw & \ctrl{5} & \ctrl{5} &\qw &\qw  &\qw&\meter{X}\\
    & \qw & \qw & \qw & \ctrl{5} & \qw  &\qw&\meter{X}\\
    & \qw & \qw & \qw &\qw &\qw  &\qw &\meter{X}\\
    & \qw & \qw & \qw &\qw &\ctrl{5} &\ctrl{5} &\meter{X}\\
    \lstick[wires=4]{Data} & \targ{} &\qw  & \qw &\qw &\qw & \qw &\qw \\
    & \qw & \targ{} & \ctrl{} &\qw &\qw & \qw &\qw\\
    & \qw & \qw & \qw &  \ctrl{} & \qw& \qw &\qw\\
    & \qw & \qw & \qw &\qw &\qw & \qw  &\qw\\
    \lstick{$|Y\rangle$}& \qw & \qw & \qw &\qw &\targ{} &\ctrl{} & \qw \rstick{$|Y\rangle$}
  \end{quantikz}
  \caption{A circuit-level implementation of a PPM for a Pauli string $P=X\otimes Y\otimes Z\otimes I$. Here $|Y\rangle$ is the $+1$ eigenstate of $Y$. The outcome of the measurement is the joint parity of the measurement outcomes. Unlike the circuit in Fig. \ref{fig:ppm_circuit}, there is no need to explicitly apply $S^{\dagger}$.\label{fig:ppm_circuit_nophase}}
\end{figure}

Therefore, we have seen that one can perform a PPM using a single $|\text{CAT}\rangle$ state and a $+1$ eigenstate of $Y$, as well as CNOT, CZ, and measurement in the $X$-basis. Once the eigenstate of $Y$ is prepared at the beginning of the computation, if the $|\text{CAT}\rangle$ state can be prepared in constant time, every operation discussed so far can be done in constant time. Next, we explain how the state preparation can be carried out, thus completing the description of our protocol.

Na\"ively, preparation of the $|\text{CAT}\rangle$ state seems to require $\mathcal{O}(d)$ time. A straightforward way to do this is to initialize all the logical qubits in the $|+\rangle$ state on a one-dimensional array and measure the joint parity of $\bar{Z}_{j}\bar{Z}_{j+1}$. While the first step requires constant time, the second step requires $\mathcal{O}(d)$ time if we insist on interacting the surface code patches only via their boundary, using lattice surgery~\cite{horsman2012surface}.

However, for our purposes, it is possible to get away with preparing the $|+\rangle$ state over a long surface code patch; see Fig.~\ref{fig:surface_cat}. Of course, this long surface code patch is not exactly the $|\text{CAT}\rangle$ state. However, they are equivalent up to a Pauli-$X$ measurement along the qubits that connect the neighboring patches. Upon measuring these qubits, one must eventually apply a correction.

\begin{figure}[h]
  \subfloat[]{\includegraphics[width=0.45\textwidth]{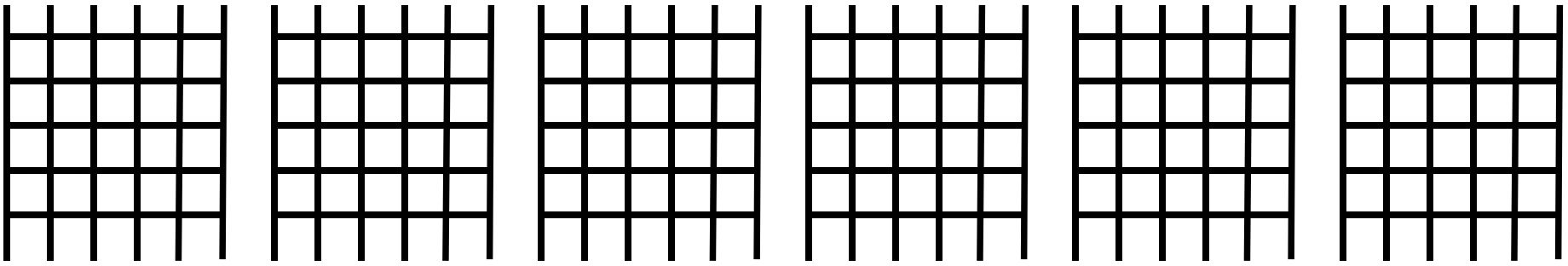}}\\
  \subfloat[]{\includegraphics[width=0.45\textwidth]{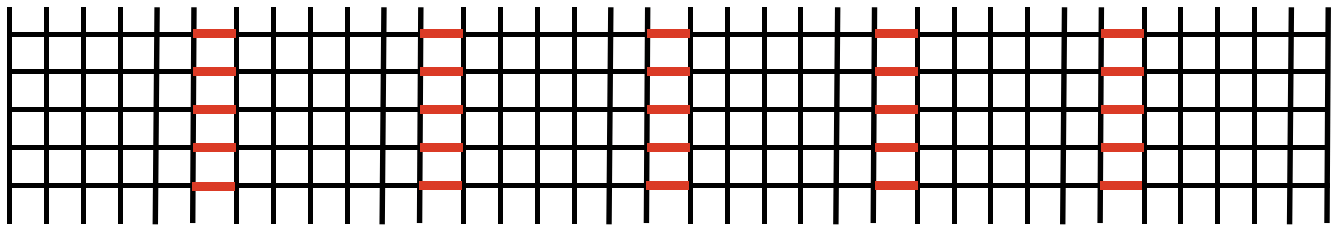}}
  \caption{(a) The data block consists of $N+1$ surface code patches. (b) The ancilla block is initially prepared on a ``long'' surface code patch. By measuring the red qubits and applying a Pauli correction, one can prepare the $|\text{CAT}\rangle$ state over the ancilla block.\label{fig:surface_cat}}
\end{figure}

These partial $X$ measurements by themselves do not seem to make the preparation of the $|\text{CAT}\rangle$ state fault-tolerant. However, used in conjunction with the final $X$-basis measurements on the remaining qubits in the ancilla block, we can measure the Pauli Product operator fault-tolerantly. This can be seen by inspecting the stabilizer generators with support strictly on the ancilla block after applying the gates described above in \ref{sec:fast_ppm}. One can verify that these include all the star operators of the long ancilla block. Moreover, while the errors will undoubtedly spread from the data block, they will spread in a local manner thanks to the transversal nature of the gates used. Therefore, the standard argument in surface code-based error correction~\cite{Dennis2002} implies that one can simply measure the entire ancilla block in the $X$-basis and decode, thus fault-tolerantly extracting the measurement outcome of the logical $X$ operator of the long ancilla block. This measurement outcome is \emph{precisely} the PPM. Since the preparation of the $|\text{CAT}\rangle$ state on the long ancilla block can be performed in $\mathcal{O}(1)$ time, by the discussion in \ref{sec:fast_ppm}, we conclude that it is possible to perform a $\mathcal{O}(1)$-time fault-tolerant PPM.

To summarize, we have shown that an arbitrary fault-tolerant PPM can be performed in constant time. Moreover, all of these operations can be implemented on a planar architecture, involving two folded layers of surface code patches that are superimposed together; see Fig.~\ref{fig:fold}.

\begin{figure}[h]
 \subfloat[]{\includegraphics[width=0.9\columnwidth]{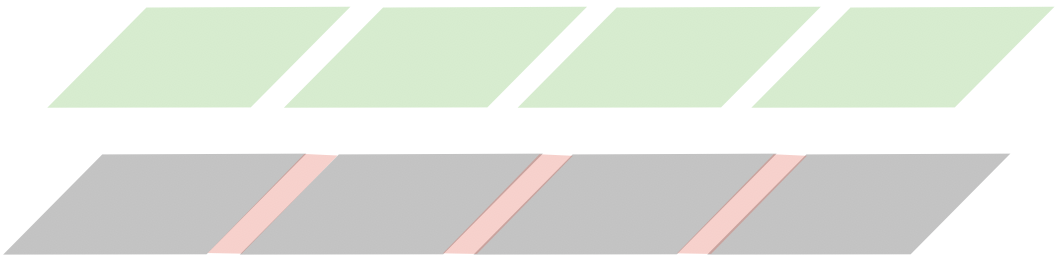}} \\
 \subfloat[]{\includegraphics[width=0.9\columnwidth]{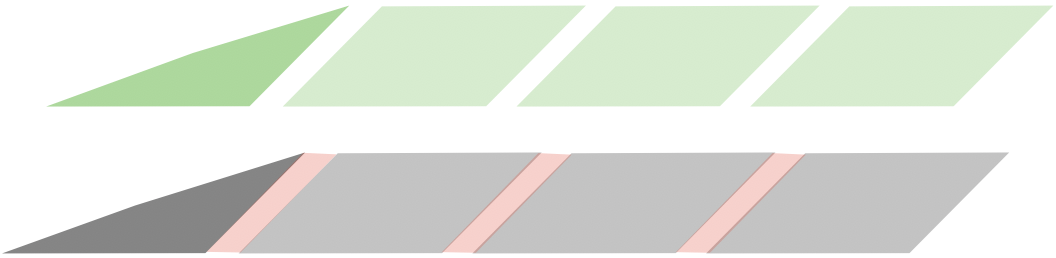}} \\
 \subfloat[]{\includegraphics[width=0.9\columnwidth]{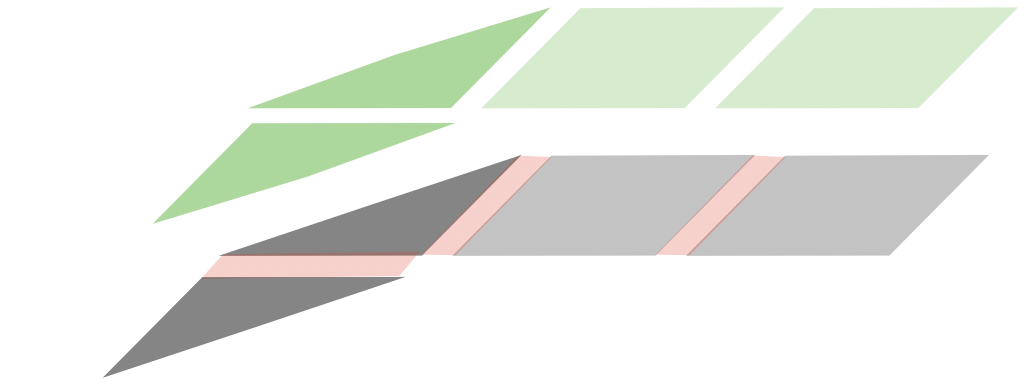}}
 \caption{(a) Two layers of arrays of surface code patches, the top being the data block and the bottom being the ancilla block. As it stands, the logical Hadamard cannot be implemented transversally. (b) and (c) By folding each surface code patch along the diagonal line sequentially, one can obtain a planar architecture in which all the requisite Hadamard and CNOT gates can be implemented transversally~\cite{Moussa2016}. After folding all the surface code patches, we obtain two layers of folded surface code patches. \label{fig:fold}}
\end{figure}

Injection of distilled magic states can be implemented by simply increasing the number of logical qubits in the data and the ancilla blocks. Specifically, suppose we can create $m$ $|T\rangle$ states in a single logical clock cycle. It will suffice to add $m$ extra logical qubits in the data and the ancilla block and store the magic states in the newly allotted qubits in the data block. Since the number of magic state factories is typically much smaller than the number of logical qubits, this additive overhead is expected to be small. Moreover, the magic state factories can be spaced out sufficiently far apart from each other to avoid any routing issues. In the regime studied in this paper, once the number of magic state factories exceeds $\sim 10$, the gain from having an extra magic state factory is expected to be insignificant. Since the smallest number of logical qubits needed exceeds a few thousand, we can take the average spacing between the magic state factories to be at least a few hundred, enough room to have independent magic state factories that do not interfere with each other. Of course, one may opt to use more sophisticated distillation protocols that generate multiple magic states at the same time~\cite{bravyihaahmagic,haah2017magic,haah2017magic_low,haah2018codes}. In those cases, the extra logical qubits will need to be physically close to each other.

One may worry that assigning a fixed location for the magic states will cause a routing problem; as we explain below, this is not the case. Thanks to the fact that we can apply arbitrary PPMs, the location of the magic states does not matter.

\subsection{Compilation and Execution} 
\label{sec:compilation}
We now explain how the newly introduced PPM can be used in practice to compile a quantum algorithm. What we describe below mostly follows from Litinski's observations~\cite{litinski2019game}. Suppose we are given a description of a quantum circuit in terms of Clifford and T-gates. One can, without loss of generality, rewrite this circuit as a sequence of gates of the following form:
\begin{equation}
  \exp\left(i\frac{\pi}{8}P\right), \label{eq:gate_ppr}
\end{equation}
where $P$ is a Pauli operator, followed by Cliffords and single-qubit measurements. Instead of applying the Cliffords directly, one can simply replace the single-qubit measurements by PPMs. Moreover, Litinski showed that a gate in Eq.~\eqref{eq:gate_ppr} can be implemented by consuming a single $|T\rangle$ state and a PPM, followed by a Clifford correction. Therefore, an arbitrary quantum computation can be expressed as a sequence of PPMs on $|T\rangle$ states and a data block, followed by a Clifford correction.

In fact, it is possible to remove the need to apply the Clifford correction entirely. The idea is to dynamically update the gates used in the circuit. Without loss of generality, suppose we have converted the circuit in the form described above, namely, as a sequence of PPMs followed by Clifford corrections of the following form:
\begin{equation}
  \exp\left( i\frac{k\pi}{4}P \right) \label{eq:gate_Clifford_correction}
\end{equation}
for some $k \in \mathbb{Z}$, where $P$ is the Pauli being measured. One can view this as a low-level instruction to the fault-tolerant quantum computer. Given such an instruction, we can execute the circuit on a quantum computer in the following way.

Define the \emph{Clifford frame}, which is a lookup table of size $2N \times (2N+1)$, where $N$ is the number of qubits in the circuit. The Clifford frame is similar to the \emph{tableau} first introduced by Aaronson and Gottesman~\cite{Aaronson2004} and is a compact representation of a Clifford gate. Since the Clifford group generates an automorphism on the Paulis, it suffices to specify their action on the generators, which we choose to be $\bar{X}_n$ and $\bar{Z}_n$ for $n=1,\ldots, N$, each representing the logical $X$- and $Z$-operators of each surface code patch. It is convenient to view the Clifford frame as a table that tells us how to convert each PPM and Clifford correction based on the Clifford corrections that precede them. Specifically, let $\widetilde{X}_n$ and $\widetilde{Z}_n$ be the $n$'th Pauli operator appearing in the instruction. The Clifford frame converts these Paulis in the following way,
\begin{equation}
  \begin{aligned}
    \widetilde{X}_n \to i^{s} \left(\bigotimes_{j=1}^N \bar{X}_j^{m^X_j} \right) \left( \bigotimes_{j=1}^N \bar{Z}_j^{m^Z_j} \right)
  \end{aligned}
\end{equation}
where $m^X_j, m^Z_j \in \{0, 1\}$ and $s\in \{0,1,2,3\}$ and similarly for $\widetilde{Z}_n$. The Clifford frame table contains this data; see Table~\ref{table:clifford_frame} for an example.

\vspace{0.25cm}

\begin{table}[h]
\renewcommand{\arraystretch}{1.4}
  \begin{tabular}{c|c|c|c|c|c|c|c}
    & $m^X_1$ & $m^X_2$ & $m^X_3$ & $m^Z_1$ & $m^Z_2$ & $m^Z_3$ & $s$ \\
    \hline
    $\widetilde{X}_1$ & 1 & 0  &0  &0  &0  &0  & $2$ \\
    \hline
    $\widetilde{X}_2$ & 1 & 1  & 0 & 1 &0  &0 &  $1$\\
    \hline
    $\widetilde{X}_3$ & 0 &0  &1  &0  &0  &0  & $0$\\
    \hline
    $\widetilde{Z}_1$ & 1 &0  &0  & 1  &0  &0  & $3$\\
    \hline
    $\widetilde{Z}_2$ & 0 & 1  &0  &0  &0  &0  & $0$ \\
    \hline
    $\widetilde{Z}_3$ & 0 &0  &0  &0  &0  &1 & $2$
  \end{tabular}
  \caption{An example of a Clifford frame for a three-qubit system.} \label{table:clifford_frame}
\end{table}

When executing a quantum circuit, the instruction will remain as is, but the Clifford frame will change dynamically. In the beginning, the Clifford frame will be initialized such that
\begin{equation}
  \begin{aligned}
    \widetilde{X}_j \to \bar{X}_j, \\
    \widetilde{Z}_j \to \bar{Z}_j.
  \end{aligned}
\end{equation}
In this very first step, we will have a PPM, which can be implemented using either the procedure described above or the one in Ref.~\cite{litinski2019game}. Then, suppose the result of the PPM indicates that a Clifford update must be carried out. Instead of applying this Clifford, say $C$, we can update all the remaining PPMs and Clifford corrections by updating the Pauli that underlies them as
\begin{equation}
  P \to C^{\dagger}P C.
\end{equation}
However, the instruction set for practical quantum algorithms will be lower bounded by the number of T-gates used, which is often on the order of at least $n_T=10^{10}$ even for state-of-the-art algorithms~\cite{DoubleFactorized_MSFT,Lee2020}. Updating every gate appearing afterwards will lead to a number of updates of order $\mathcal{O}(n_T^2) \approx 10^{20}$, which is a significant amount of computation even for a modern computer. Therefore, this is an unwieldy approach. Instead, it is better to update the table such as in Table~\ref{table:clifford_frame}. This amounts to updating $\widetilde{X}_n \to C^{\dagger}\widetilde{X}_n C$ and $\widetilde{Z}_n \to C^{\dagger} \widetilde{Z}_n C$ for every $n$.

Happily, such a calculation is extremely simple. For instance, given a $\widetilde{X}_n$, if it commutes with the Pauli operator $P$ appearing in the correcting Clifford defined in Eq.~\eqref{eq:gate_Clifford_correction}, no update is needed. Checking the commutation relation only requires $\mathcal{O}(1)$ number of bitwise operations and parity calculations for $N$-bit strings. If they do not commute, $\widetilde{X}_n$ is updated as $\widetilde{X}_n P$ up to a phase. Both the phase and the updated Pauli can also be computed using $\mathcal{O}(1)$ number of bitwise operations and parity calculations for $N$-bit strings. Thus, the amount of classical computation scales as $\mathcal{O}(N)$ parity calculations and bitwise-operations over $N$-bit strings.

Let us emphasize that this calculation can be trivially parallelized to $\mathcal{O}(1)$ parity calculations and bit-wise operations distributed over $2N$ processors. Moreover, modern CPUs can perform such operations over thousands of bits on the order of \;nanoseconds. Therefore, we do not expect the update of the Clifford frame to be a significant bottleneck for the computation.

\subsection{Speedups} 
\label{sec:speedups}
At first, one may think that a sequence of PPMs can be implemented in $\mathcal{O}(1)$ time, independent of whether they commute with each other. While we cannot completely discount this possibility, we will demand the sequence of PPMs to satisfy an important constraint. If we were to implement $k$ PPMs in a single logical clock cycle, we demand all the Paulis that underlie the PPMs to commute with each other.

If they do not, we encounter the following difficulty. Decoding a logical qubit requires syndrome measurement outcomes in a $d\times d\times d$ block, where the last factor of $d$ comes from time~\cite{Dennis2002}. Upon collecting all of these measurement outcomes and decoding, one must apply a correction operation. If certain layers of the $d\times d\times d$ block undergo the transversal gates described above, the measurement outcome of the ancilla block may, depending on the outcome of the decoding, have to be retroactively changed. Doing so would necessitate a change in the PPM outcome that occurred in the past, which would in turn signal that a Clifford correction in Eq.~\eqref{eq:gate_Clifford_correction} was applied incorrectly, resulting in an incorrectly-updated Clifford frame. The incorrect Clifford frame may have triggered an incorrect PPM, which may have already happened.

Fortunately, such a problem is moot if the Paulis that underlie the $k$ PPMs commute with each other. In that case, the Clifford corrections also commute with the PPMs; see Eq.~\eqref{eq:gate_Clifford_correction}. Therefore, the Clifford corrections do not change the PPMs within that block. One can simply perform all the PPMs and then update the Clifford frame before we move to the next $d\times d\times d$ block.

With this constraint in mind, the speedup we can achieve using this approach is limited by the ratio of T-count and T-depth. In our setup we can optimistically expect up to a $15$-fold speedup with a negligible change in the footprint of the device.

Another constraint is the code distance. At best, one can apply $\mathcal{O}(d)$ PPMs in a single logical clock cycle using our method. The estimated code distance in Section \ref{sec:overhead} ranged from $24$ to $44$, suggesting that the best-case scenario of a $15$-fold speedup may be attainable. However, note that whether this is possible or not depends on many microscopic details, such as the timescale needed for different physical operations. Moreover, we would like to emphasize that the threshold and the sub-threshold behavior of the code may change when we implement the transversal gates. The exact extent to which we can speed up the computation using our approach can only be determined by carefully inspecting all of these different factors. These studies are beyond the scope of this paper and are left for future work.

Keeping these caveats in mind, we can estimate the optimistic computation time for a single phase estimation for the molecules described below to be on the order of $1\sim 4$ hours for the cc-pVDZ basis set and $35\sim 90$ hours for the cc-pVTZ basis set.

\section{Proof for logarithmic depth ``Gizens" rotations}
\label{sec:gizens_proof}
\begin{thm}
  Let $N=2^n$. Define a $\ell$-layer circuit as
  \begin{equation}
    \tilde{V} = \prod_{j=1}^{\ell} \prod_{k=0}^{2^j - 2} \tilde{V}_{k2^{n-j},(k+1)2^{n-j}}(\theta_{jk}).
  \end{equation}
  Then the angles $\theta_{jk}$ can be chosen such that $\tilde{V}\gamma_0 = \gamma_{\vec{u}}$.
\end{thm}
\begin{proof}
  Proceed by induction. After $s$ layers, assume that the circuit applies
  \begin{equation}
    \gamma_0 \rightarrow \sum_{i=0}^{2^s -1} c_{i,s} \gamma_{i 2^{n-s}},
  \end{equation}
  where
  \begin{equation}
    c_{i,s} = \sqrt{\sum_{j=i2^{n-s}}^{(i+1)2^{n-s}-1} u^2_j}.
  \end{equation}
  Then the action of layer $(s+1)$ is to rotate
  \begin{align}
    \gamma_{i2^{n-s}} &\rightarrow \sin(\theta_{is})\gamma_{i2^{n-s}}+\cos(\theta_{is})\gamma_{i2^{n-s} + 2^{n-s-1}} \nonumber \\
    &= \sin(\theta_{is})\gamma_{(2i)2^{n-(s+1)}}+\cos(\theta_{is})\gamma_{(2i+1)2^{n-(s+1)}}.
  \end{align}
  Provided that we choose $\theta_{is}$ such that $\arcsin(\theta_{is}) = \frac{c_{i,s+1}}{c_{i,s}}$, we therefore have that after $(s+1)$ layers,
  \begin{equation}
    \gamma_0 \rightarrow \sum_{i=0}^{2^{s+1}-1} c_{i,s+1} \gamma_{i2^{n-(s+1)}}.
  \end{equation}
  For the base case $s=1$, we require that
  \begin{equation}
    \gamma_0 \rightarrow c_{0,1} \gamma_0 + c_{1,1} \gamma_{2^{n-1}}
  \end{equation}
  which is trivially achieved by the single Gizens rotation $\tilde{V}_{0, 2^{n-1}}(\theta)$ with $\theta = \arcsin{(c_{0,1})}$.
  After $n$ layers, the coefficients are
  \begin{equation}
    c_{i,n} = \sqrt{\sum_{j=i2^{n-n}}^{(i+1)2^{n-n}-1} u_j^2} = \sqrt{\sum_{j=i}^{i} u_j^2} = u_i^2
  \end{equation}
  and we have prepared $\gamma_{\vec{u}}$ as intended.
\end{proof}

\end{document}